\documentclass[conference]{IEEEtran}
\IEEEoverridecommandlockouts
\usepackage{cite}
\usepackage{amsmath,amssymb,amsfonts}
\usepackage{graphicx}
\usepackage{textcomp}
\usepackage{xcolor}
\usepackage{multirow}
\def\BibTeX{{\rm B\kern-.05em{\sc i\kern-.025em b}\kern-.08em
    T\kern-.1667em\lower.7ex\hbox{E}\kern-.125emX}}
\usepackage{verbatim}

\usepackage{mathrsfs}
\usepackage{amsthm}
\usepackage{diagbox}
\usepackage{verbatim}
\usepackage{dsfont}
\usepackage{bm}
\usepackage{subeqnarray}
\usepackage{cases}

\usepackage{caption}

\newtheorem{theorem}{Theorem}
\newtheorem{lemma}{Lemma}

\theoremstyle{definition}
\newtheorem{remark}{Remark}
\newtheorem{definition}{Definition}
\newtheorem{example}{Example}

\usepackage{algorithm}
\usepackage{algorithmic}
\usepackage{bbm}
\usepackage{hyperref}
\hypersetup{colorlinks = true
            linkcolor=red,
            anchorcolor=blue,
            citecolor=green}
              
\newcommand{\MF}{\mathcal{F}}
\newcommand{\MI}{\mathcal{I}}
\newcommand{\MR}{\mathcal{R}}
\newcommand{\mM}{\mathcal{M}}
\newcommand{\FF}{\mathbb{F}}

\begin{document}

\title{On the Weight Spectrum of Rate-Compatible Polar Codes}

 \author{
   \IEEEauthorblockN{Zicheng Ye\IEEEauthorrefmark{2}\IEEEauthorrefmark{3}, 
   Yuan Li\IEEEauthorrefmark{1}, 
   Zhichao Liu\IEEEauthorrefmark{1},
   Huazi Zhang\IEEEauthorrefmark{1}, 
   Jun Wang\IEEEauthorrefmark{1}, 
   Guiying Yan\IEEEauthorrefmark{2}\IEEEauthorrefmark{3}, 
   and Zhiming Ma\IEEEauthorrefmark{2}\IEEEauthorrefmark{3} }
  \IEEEauthorblockA{\IEEEauthorrefmark{2}
                     School of Mathematical Sciences, University of Chinese Academy of Sciences}
   \IEEEauthorblockA{\IEEEauthorrefmark{3}
                     Academy of Mathematics and Systems Science, CAS }
  \IEEEauthorblockA{\IEEEauthorrefmark{1}
                     Huawei Technologies Co. Ltd.}
                     
    Email: yezicheng@amss.ac.cn, \{liyuan29, liuzhichao28, zhanghuazi, justin.wangjun\}@huawei.com,  \\
    yangy@amss.ac.cn, mazm@amt.ac.cn  
\thanks{This work was supported by the National Key R\&D Program of China (2023YFA1009602).}
    }

\maketitle

\begin{abstract}
The weight spectrum plays a crucial role in the performance of error-correcting codes. Despite substantial theoretical exploration of polar codes with mother code length, a framework for the weight spectrum of rate-compatible polar codes remains elusive. In this paper, we address this gap by presenting the theoretical results for enumerating the number of minimum-weight codewords for quasi-uniform punctured, Wang-Liu shortened, and bit-reversal shortened decreasing polar codes. Additionally, we propose efficient algorithms for computing the average spectrum of random upper-triangular pre-transformed shortened and punctured polar codes. Notably, our algorithms operate with polynomial complexity relative to the code length. Simulation results affirm that our findings yield a precise estimation of the performance of rate-compatible polar codes.

\end{abstract}

\section{Introduction}
Polar codes \cite{Arikan2009}, introduced by Ar{\i}kan, are a significant breakthrough in coding theory. As the code length approaches infinity, polar codes can approach channel capacity under successive cancellation (SC) decoding. For short to moderate block lengths, successive cancellation list (SCL) decoding \cite{Niu2012, Tal2015} can significantly improve the error-correcting performance. With a sufficiently large list size, SCL decoding algorithms can approach the performance of maximum likelihood (ML) decoding.

The inherent structure of polar codes, which relies on the Kronecker product, restricts the original code length to powers of two. However, practical applications often require transmitting messages with varying code lengths. To address this issue and offer the necessary flexibility in code length, puncturing and shortening techniques are developed for polar codes.

To optimize the performance of punctured and shortened polar codes, several rate-matching patterns were developed, including the quasi-uniform puncturing (QUP) \cite{Niu2013} and Wang-Liu shortening \cite{Wang2014}, which are designed to ensure that the puncturing or shortening positions follow a quasi-uniform distribution after bit-reversal permutation. Recently, the QUP and Wang-Liu shortened polar codes were proved to achieve channel capacity \cite{Shuval2024}. The bit-reversal shortening scheme was proposed in \cite{Bioglio2017}, this method leverages a fixed reliability sequence for construction, providing a convenient alternative for implementing shortened polar codes. The exploration of puncturing and shortening techniques is extensive, with numerous researches \cite{Niu2016, Miloslavskaya2015, Chandesris2017, Oliveira2018, Han2022, Li2023Two} to offer insights and advancements in the design of rate-compatible polar codes.

The weight spectrum has a significant impact on the ML decoding performance, which can be approximated through the union bound based on the number of low-weight codewords \cite{Sason2006}. However, in the general case, the complexity of computing the exact weight spectrum grows exponentially with the code length. Various efforts have been made to analyse the weight spectrum of polar codes. The authors in \cite{Bardet2016} regarded polar codes as decreasing monomial codes and applied lower triangular affine (LTA) automorphisms to determine the number of codewords with the minimum weight $w_{\min}$. Building upon this foundation, recent studies quantified the number of codewords with weights of $1.5w_{\min}$ \cite{Rowshan2023} and those with weights below $2w_{\min}$ \cite{Ye2023}.

Algorithmic approaches for approximating the weight spectrum of polar codes are developed. For instance, \cite{Li2012} proposed to employ SCL decoding with a large list size at high signal-to-noise ratios (SNR) to collect low-weight codewords. Subsequently, \cite{Liu2014} refined this method by optimizing memory usage. In parallel, probabilistic approaches with polynomial complexity were introduced in \cite{Valipour2013, Zhang2017} to approximate the weight spectrum of polar codes to strike a balance between accuracy and computational efficiency. 11

To improve the weight spectrum of polar codes, pre-transformed polar codes \cite{Li2019}, including CRC-Aided (CA) polar codes\cite{Niu2012}, parity-check (PC) polar codes\cite{Zhang2018}, and polarization-adjusted convolutional (PAC) codes\cite{Arıkan2019}, were proposed. It has been proved that such upper-triangular pre-transformations do not decrease the code distance \cite{Li2019}. The number of minimum-weight codewords for pre-transformed polar codes was calculated in \cite{Rowshan2021}. Building on this foundation,  \cite{Zunker2024} proposed a tree intersection method to reduce the computational complexity. Furthermore, the authors in \cite{Yao2023} presented a simplified algorithm to calculate the exact weight spectrum of both original and specific pre-transformed polar codes using coset analysis with exponential complexity. The method was extended to punctured and shortened pre-transformed polar codes in \cite{Ellouze2024}. The average weight spectrum of CA-polar codes, based on the input-output weight enumerating function of original polar codes, was researched in \cite{Ricciutelli2019}. However, determining the input-output weight enumerating function for original polar codes remains challenging. In \cite{Li2021, Li2023}, the authors proposed efficient recursive formulas to calculate the average weight spectrum of pre-transformed polar codes with polynomial complexity. 

Moreover, a method for computing the weight spectrum of rate-compatible polar codes was proposed in \cite{Miloslavskaya2022}. The authors introduced a technique to enumerate all codewords up to a certain Hamming weight for general binary linear block codes. The complexity of using this method is determined by the number of low-weight codewords in the component codes. Although the complexity for polar codes is lower than that of general linear codes, it is still super-polynomial in the code length.

In this paper, we compute the number of  minimum-weight codewords for QUP, Wang-Liu shortened, and bit-reversal shortened decreasing polar codes. We also provide iterative formulas for calculating the average weight spectrum of random pre-transformed rate-compatible polar codes. The algorithms operate with polynomial complexity in the code length. This computational characteristic ensures that our methods remain scalable and practical for use with polar codes of varying lengths.

The rest of this paper is organized as follows. Section II offers a comprehensive review of the concepts and previous work. In Section III, we introduce formulas and algorithms to enumerate the minimum-weight codewords for bit-reversal shortened polar codes. Section IV focuses on enumerating the minimum-weight codewords for QUP and Wang-Liu shortened polar codes, and presents the associated formulas and algorithms. In Section V, we calculate the average spectrum of pre-transformed rate-compatible polar codes. In Section VI, we showcase the numerical results obtained through the application of our formulas and algorithms. Finally, we conclude the paper by summarizing the key findings and contributions in Section VII.

\section{Preliminaries}

\subsection{Polar codes as monomial codes}
Let $\bm{F}=\begin{bmatrix} 1&0 \\ 1&1 \end{bmatrix}$, and $m$ is a positive integer, $N=2^m$ and $\bm{F}_N=\bm{F}^{\otimes m}$, where $\otimes$ is Kronecker product:
$$
\bm{A} \otimes \bm{B} =\left[\begin{array}{ccc}
A_{11} \bm{B} & \cdots & A_{1 k} \bm{B} \\
\vdots & \ddots & \vdots \\
A_{t 1} \bm{B} & \cdots & A_{t k} \bm{B}
\end{array}\right].
$$
Polar codes can be generated by $K$ rows of $\bm{F}_N$ selected as the set of information bits $\mathcal{I}$. $\MF = \MI^c$ is called the frozen set. We denote the polar code with information set $\MI$ as $C(\MI) = \{\bm{c} = \bm{u}\bm{F}_N\mid \bm{u}_{\MI^c}=0\}$.

Polar codes can be described as monomial codes in the ring $\mathbb{F}_2[x_1,\cdots, x_m]/(x^2_1-x_1,\cdots, x^2_m-x_m)$ \cite{Bardet2016}. The monomial set is denoted by
$$
\mM \triangleq \{x_1^{a_1}...x_m^{a_m}\mid (a_1,...,a_{m})\in\FF_2^m\}.
$$
Let $f = x_1^{a_1}...x_{m}^{a_m}$ be a monomial in $\mM$, the degree of $f$ is defined as $\text{deg}(f)$, i.e., the Hamming weight of $(a_1,...,a_{m})$. In particular, 1 is a monomial with degree zero. A monomial $f$ is a factor of a monomial $g$ if $g$ can be written as the product of $f$ and another monomial $h$, i.e., $g=fh$.

The polynomial set is denoted by
$$
\MR_{\mM} \triangleq \{\sum_{f\in \mM} u_f f \mid u_f\in\FF_2\},
$$
and the degree of the polynomial $p =\sum_{e\in \mM} u_f f$ is defined as $\text{deg}(p) = \max_{u_f\neq 0} \{\text{deg}(f)\}$. A polynomial $p$ is said to be linear if deg$(p) = 1$. 

For each $z\in \{1,2,...,2^m\}$, there is a unique binary representation $\bm{a}=(a_1,...,a_m)\in\FF_2^m$, where $a_1$ is the least significant bit, such that
\[z = D(\bm{a}) \triangleq \sum_{i=1}^{m} 2^{i-1} (a_i \oplus 1) + 1,\]
where $\oplus$ is the mod-$2$ sum in $\FF_2$.

Denote $p(\bm{a})$ to be the evaluation of polynomial $p$ at point $\bm{a}\in\FF_2^m$. The length-$N$ evaluation vector of $p\in \MR_{\mM}$ is denoted by
$$
\text{ev}(p) \triangleq (p(\bm{a}))_{\bm{a}\in\FF_2^m},
$$
where the element $p(\bm{a})$ is placed at the $D(\bm{a})$-$th$ position in $\text{ev}(p)$. 

Denote $f_i = x_1^{a_1}...x_m^{a_m}$ with $i = D(a_1,...,a_m)$. Then $\text{ev}(f_i)$ is exactly the $i$-$th$ row of $\bm{F}_N$. Therefore, for a codeword $\bm{c} = \bm{u}\bm{F}_N \in C(\MI)$, we have $\bm{c} = \text{ev}(p)$ where $p = \sum_{i=1}^N u_i f_i$. In this paper, we sometimes use $p$ to represent codewords $\bm{c} = \text{ev}(p)$. The Hamming weight of $p$ is $\text{wt}(p) \triangleq \text{wt}(\bm{c})$, and $p_{[a,b]} \triangleq \bm{c}_a^b = (c_a,c_{a+1},\dots,c_b)$.

The information set $\MI$ of polar codes can be regarded as a subset of the monomial set $\mM$. 

\begin{definition}
Let $\MI$ be a set of monomials. The monomial code $C(\MI)$ with code length $N = 2^m$ is defined as
$$
C(\MI) \triangleq \operatorname{span} \{\operatorname{ev}(f): f \in \MI\}.
$$
\end{definition}

If the maximum degree of monomials in $\MI$ is $r$, we call the monomial code $C(\MI)$ is $r$-$th$ order.

\begin{example}

The example shows the row vector representations of $\bm{F}_8$.

$$
\begin{array}{lc}
\begin{array}{c} x_1= \\ x_2=\\ x_3 =\\ \hline x_1x_2x_3 \\ x_2x_3 \\ x_1x_3 \\ x_3 \\ x_1x_2\\ x_2 \\ x_1 \\ 1
\end{array}
\begin{array}{cccccccc}
1 &0 &1 &0 &1 &0 &1 &0\\
1 &1 &0 &0 &1 &1 &0 &0\\
1 &1 &1 &1 &0 &0 &0 &0 \\
\hline
1 & 0 & 0 & 0 & 0 & 0 & 0 & 0\\
1 & 1 & 0 & 0 & 0 & 0 & 0 & 0\\
1 & 0 & 1 & 0 & 0 & 0 & 0 & 0\\
1 & 1 & 1 & 1 & 0 & 0 & 0 & 0\\
1 & 0 & 0 & 0 & 1 & 0 & 0 & 0\\
1 & 1 & 0 & 0 & 1 & 1 & 0 & 0\\
1 & 0 & 1 & 0 & 1 & 0 & 1 & 0\\
1 & 1 & 1 & 1 & 1 & 1 & 1 & 1
\end{array}
\end{array}
$$

If the information set is $\MI=\{1,x_1,x_2,x_3\}$, then the generator matrix of the polar code $C(\MI)$ is

$$
\begin{bmatrix}
1 & 1 & 1 & 1 & 0 & 0 & 0 & 0\\
1 & 1 & 0 & 0 & 1 & 1 & 0 & 0\\
1 & 0 & 1 & 0 & 1 & 0 & 1 & 0\\
1 & 1 & 1 & 1 & 1 & 1 & 1 & 1
\end{bmatrix}
$$

The polynomial $p = x_3 + x_1 + 1$ is a codeword in $(C(\MI))$, which represents $\text{ev}(p) = (1,0,1,0,1,0,0,1)$.

\end{example}

\subsection{Decreasing monomial codes}

The partial order of monomials was defined in \cite{Bardet2016} and  \cite{Schurch2016}. Two monomials with the same degree are ordered as $x_{i_1}...x_{i_t}\preccurlyeq x_{j_1}...x_{j_t}$ if and only if $i_l \leq j_l$ for all $l\in\{1,...,t\}$, where we assume $i_1 <...< i_t$ and $j_1 <...< j_t$. This partial order is extended to monomials with different degrees through divisibility, namely $f \preccurlyeq g$ if and only if there is a factor $g'$ of $g$ such that $f \preccurlyeq g'$. In other words, $x_{i_1}...x_{i_t}\preccurlyeq x_{j_1}...x_{j_s}$ if and only if $t\leq s$ and $i_{t-l} \leq j_{s-l}$ for all $l\in\{0,...,t-1\}$.

An information set $\MI\subseteq \mathcal{M}$ is decreasing if $\forall f\preccurlyeq g$ and $g\in\MI$ we have $f \in \MI$. A decreasing monomial code $C(\MI)$ is a monomial code with a decreasing information set $\MI$. If the information set is selected according to the Bhatacharryya parameter or the polarization weight (PW) method \cite{He2017}, the polar codes are decreasing. In fact, $f \preccurlyeq g$ means $f$ is universally more reliable than $g$, so we only focus on decreasing monomial codes.

\subsection{Pre-Transformed Polar Codes}
Define
$$
\bm{T}
=\begin{bmatrix}
1  &  T_{12}  & \cdots\ & T_{1N}\\
0  &  1  & \cdots\ & T_{2N}\\
 \vdots   & \vdots & \ddots  & \vdots  \\
 0 & 0  & \cdots\ & 1\\
\end{bmatrix}
$$
be an $N\times N$ upper-triangular pre-transformation matrix. Let $\bm{G}_N = \bm{T} \bm{F}_N$ be the generator matrix of pre-transformed polar code. The codeword of pre-transformed polar code is given by $\bm{c}=\bm{u} \bm{G}_N=\bm{uT}\bm{F}_N$, where $\bm{u}_{\mathcal{F}}=\textbf{0}$.

\subsection{minimum-weight codewords of decreasing monomial codes} 

For any code $C$, define $A_d(C)$ to be the number of codewords with weight $d$ in $C$. Let $C(\MI)$ be an $r$-th order decreasing monomial code, $f_t = x_{i_1}\cdots x_{i_r} = x_1^{a_1}...x_m^{a_m}$ with $t=D(a_1,\cdots,a_m)$. Define $T_{\MF}(f_t) = \{p = \sum_{i=1}^N  u_i f_i \mid  \bm{u}_1^{t-1} = \bm{0}, u_t = 1, \bm{u}_{\mathcal{F}} =\bm{0}, \text{wt}(p) = \text{wt}(f_t)\}$, i.e., $T_{\MF}(f_t)$ consists of codewords  with the same weight as $f_t$ in the coset where $u_t$ is the first non-zero bit in $u_1^N$. From \cite{Bardet2016}, the complete minimum-weight codewords of $C(\MI)$ are $\bigcup_{f\in\MI_r}T_{\MF}(f)$, where $\MI_r$ is the subset of degree-$r$ monomials in $\MI$. The following theorem presents the minimum-weight codewords in decreasing monomial codes.

\begin{theorem}\cite{Bardet2016}\label{minweight2016}
Let $C(\MI)$ be an $r$-th order decreasing monomial code. Then $T_{\MF}(x_{i_1}\cdots x_{i_r})$ consists of the codewords
\begin{equation}\label{eq:min_weight}
\prod_{j=1}^r (x_{i_j} + \sum_{k\in B(f,j)} a_{i_j, k}x_k + a_{i_j,0}),
\end{equation}
where $B(f,j)$ is the set $\{1\leq k\leq m \mid k<i_j, k\neq i_1,\dots, i_r\}$. 

Define $\lambda_{x_{i_1}\cdots x_{i_r}} = 2^{\sum_{t=1}^r (i_t-t+1)}$, then the number of codewords in $T_{\MF}(x_{i_1}\cdots x_{i_r})$ is $\lambda_{x_{i_1}\cdots x_{i_r}}$. Furthermore, 
\begin{equation}
A_{2^{m-r}}(C(\MI)) = \sum_{f \in \MI_r} \lambda_f.
\end{equation}
\end{theorem}

From Theorem \ref{minweight2016}, it is clear that $T_{\MF}(f)$ is not relevant to $\MF$ as long as $C(\MI)$ is decreasing, so we abbreviate it as $T(f)$.

\subsection{Automorphism group of decreasing monomial codes}

Let $C$ be a code with length $N$. A permutation $\pi$ in the symmetric group $S_N$ is an automorphism of $C$ if for any codeword ${\bm c} =(c_1,\dots,c_N)\in C$, 
$$
\pi({\bm c})  \triangleq (c_{\pi(1)},\dots,c_{\pi(N)})\in C.
$$
The automorphism group Aut$(C)$ is the subgroup of $S_N$ containing all the automorphisms of $C$. 

Let ${\bm A}$ be an $m\times m$ binary invertible matrix and ${\bm b}$ be a length-$m$ binary column vector. The affine transformation $({\bm A},{\bm b})$ permutes $D(\bm{z})$ to $D(\bm{A}\bm{z}+\bm{b})$ for $\bm{z}\in \FF_2^m$. The affine transformation group $\text{GA}(m)$ is the group consisting of all the affine transformations. 

A matrix $\bm{A}$ is lower-triangular if $A_{i,i}=1$ and $A_{i,j}=0$ for all $j>i$. The group LTA$(m)$ consists of all the affine transformations $({\bm A},{\bm b})$ where $\bm{A}$ is lower-triangular. Similarly, a matrix $\bm{A}$ is upper-triangular if $A_{i,i}=1$ and $A_{i,j}=0$ for all $j<i$.

Let $\bm{s}=(s_1,\dots,s_l)$ be a positive integer vector with $s_1+\dots+s_l = m$. BLTA$(\bm{s})$ is a BLTA group consisting of all the affine transformations  $({\bm A},{\bm b})$ where ${\bm A}$ can be written as a block matrix of the following form
\begin{equation}
\begin{bmatrix}
\bm{B}_{1,1} & \bm{0} & \cdots & \bm{0} \\
\bm{B}_{2,1} & \bm{B}_{2,2} & \cdots & \bm{0} \\
\vdots & \vdots  & \ddots & \vdots \\
\bm{B}_{l,1} & \bm{B}_{l,2} & \cdots & \bm{B}_{l,l}
\end{bmatrix}, \label{eqma}
\end{equation}
where $\bm{B}_{i,i}$ are full-rank $s_i\times s_i$ matrices.

In \cite{Bardet2016}, it is proved that LTA$(m)$ belongs to the automorphism group of a decreasing monomial code. Later, a larger BLTA group is known to be the affine automorphism group of any specific decreasing monomial code in \cite{Geiselhart2021, Li2021complete}.

\begin{example}

When $m = 3$, let $\bm{A} = \bm{I}_3$ be the identity matrix and $\bm{b} = (1,1,1)$. Then $(\bm{A}, \bm{b})$ transforms $x_i$ to $x_i+1$, which means it interchanges the $i=D((x_1,x_2,x_3))$-$th$ bit and the $(N+1-i) = D((x_1+1,x_2+1,x_3+1))$-$th$ bit for all $1\leq i\leq N$. The corresponding permutation is $(8,7,6,5,4,3,2,1)$.

Let $\bm{c} = (0,0,1,1,0,1,1,0)$ be a codeword in a decreasing monomial code $C(\MI)$. Since $(\bm{I}_3, \bm{b})$ with $\bm{b} = (1,1,1)$ is a lower-triangular affine transformation, it is an automorphism of $C(\MI)$. In fact, it permutes $\bm{c} = (0,0,1,1,0,1,1,0)$ to another codeword $\bm{c}' = (0,1,1,0,1,1,0,0)\in C(\MI)$.
\end{example}

\subsection{Puncturing and shortening} 

Denote $[N] = \{1, 2, \cdots, N\}$, the difference set of two sets $X$ and $Y$ is denoted as $X/Y$, i.e., $X/Y = \{i\in X\mid i\notin Y\}$. For a codeword $\bm{c}\in\FF_2^N$ and set $X\subseteq [N]$, $\bm{c}_X$ is the codeword $\bm{c}$ restricted on $X$, i.e., $\bm{c}_X = (c_i)_{i\in X}$.

\begin{definition}[Punctured Code and Puncturing Pattern]
Let $C(\MI)$ be a polar code with length $N=2^m$, and $X\subseteq [N]$ be a set called puncturing pattern. The punctured code of $C(\MI)$ with puncturing pattern $X$ is $C_P(\MI,X) = \{\bm{c}_{[N]/X} \mid \bm{c}\in C(\MI)\}$.
\end{definition}

\begin{definition}[Shortened Code and Shortening Pattern]
Let $C(\MI)$ be a polar code with length $N=2^m$, and $Y\subseteq [N]$ be a set called by shortening pattern. The shortened code of $C(\MI)$ with shortening pattern $Y$ is $C_S(\MI,Y) = \{\bm{c}_{[N]/Y} \mid \bm{c}\in C(\MI), \bm{c}_{Y} = \bm{0}\}$.
\end{definition}

Denote the quasi-uniform (QU) sequence to be $\bm{q}=(1,2,\dots,N)$. The QUP pattern $X_i$ \cite{Niu2013} is the set of the first $i$ bits in $\bm{q}$, i.e., $\{1,2,\dots,i\}$. The QUP polar code with length $N-i$ and information set $\MI$ is denoted by $C_P(\MI,X_i)$. The Wang-Liu shortening pattern $Y_i$ \cite{Wang2014} is the set of the last $i$ bits in $\bm{q}$, i.e., $\{N-i+1,N-i+2,\dots,N\}$. The Wang-Liu shortened polar code with length $N-i$ and information set $\MI$ is denoted as $C_S(\MI,Y_i)$. 

Denote the bit-reversal sequence $\bm{q}'$ to be the bit reversal of $\bm{q}$, i.e., $ \bm{q}'_{D(a_m,\dots,a_1)} =  \bm{q}_{D(a_1,\dots,a_m)} = D(a_1,\dots,a_m)$. The bit-reversal shortening pattern $Y'_i$ is the last $i$ bits of $\bm{q}'$, and the bit-reversal shortened polar code with length $N-i$ and information set $\MI$ is denoted as $C_S(\MI,Y_i')$.

\begin{example}
When $m=3$, the bit-reversal sequence $\bm{q}' = (1,5,3,7,2,6,4,8)$ is the bit reversal of  $\bm{q}=(1,2,3,4,5,6,7,8)$. To see this, for example, $\bm{q}'_4 = \bm{q}'_{D(0,0,1)} =  \bm{q}_{D(1,0,0)} = \bm{q}_7 = 7$. Hence, in the process of bit-reversal shortening, the 8-$th$ bit is the initial one to be shortened, followed by the 4-$th$ bit as the second, and the sequence continues in this manner.
\end{example}

\begin{remark}\label{remark1}
Note that when the puncturing pattern $X$ complies with binary domination \cite{Jang2019}, the symmetric capacity of the monomial $x_1^{a_1}\cdots x_m^{a_m}$ with $D(a_1,\dots ,a_m)\in X$ becomes zero. When the shortening pattern $Y$ complies with binary domination, the coefficient of monomial $x_1^{a_1}\cdots x_m^{a_m}$ with $D(a_1,\dots ,a_m)\in Y$ is frozen to zero. Note that QUP, Wang-Liu shortening, and bit-reversal shortening patterns all comply with binary domination. Consequently, for a polynomial $p$ in one of these codes, the coefficient of any monomial $x_1^{a_1}\cdots x_m^{a_m}$ in $p$ with the corresponding index $D(a_1,\dots ,a_m)$ that is in the puncturing pattern $X$ or the shortening pattern $Y$ is zero. Therefore, with a slight abuse of notation, when the $D(a_1,\dots ,a_m)$-$th$ bit is punctured or shortened, we also say the monomial  $x_1^{a_1}\cdots x_m^{a_m}$ is punctured or shortened. 
\end{remark}

In this paper, we focus on the minimum-weight codewords of decreasing QUP, Wang-Liu shortened, and bit-reversal shortened polar codes. We say a shortened polar code is decreasing if the union of the information set and the shortened set is decreasing. 

\subsection{minimum-weight codewords of punctured and shortened polar codes} 

In Section III and IV, we calculate the number of minimum-weight codewords for QUP, Wang-Liu shortened, and bit-reversal shortened polar codes. We provide an outline of the proof in this subsection. Define $T_P(f,X) = \{p_{[N]/X} \mid p\in T(f)\}$, which is the set of punctured codewords in $T(f)$ with the puncturing pattern $X$.  Similarly, define $T_S(f,Y) = \{p_{[N]/Y} \mid p\in T(f), p_Y = \bm{0}\}$,  which is the set of codewords in $T(f)$ with the shortening pattern $Y$. For $f = x_1^{a_1}\cdots x_m^{a_m}$, we only need to study $T_P(f,X)$ (resp. $T_S(f,Y)$) when $D(a_1,\dots,a_m)\notin X$ (resp. $Y$) due to Remark \ref{remark1}.

In Section III, we calculate the number of minimum-weight codewords for bit-reversal shortened polar codes $C_S(\MI, Y'_i)$. $T_S(f,Y'_i)$ is a subset of $T(f)$ in which the coefficients of monomials in $Y'_i$ are 0. As shortening does not change the codeword weight, the set of minimum-weight codewords in $C_S(\MI, Y'_i)$ is $\bigcup_{f\in\MI_r}T_S(f,Y'_i)$. We demonstrate that $T_S(f,Y'_{i-1})/T_S(f,Y'_i)$ is not empty if and only if the $i$-$th$ shortened monomial is a factor of $f$. In this case,  the number of polynomials removed is $|T_S(f,Y'_{i-1})/T_S(f,Y'_i)| = \lambda_f/2^r$. As $|T_S(f,\varnothing)| = |T(f)| = \lambda_f$, $|T_S(f,Y'_i)|$ can be calculated recursively based on the number of shortened bits.

In Section IV, we analyse the number of minimum-weight codewords for QUP and Wang-Liu shortened polar codes. For QUP, define $M_P(\MI,X_i)$ to be the set of minimum-weight codewords in $C_P(\MI,X_i)$. We prove that a codeword in $M_P(\MI,X_i)$ must also be a minimum-weight codeword before puncturing, meaning it belongs to $\bigcup_{f\in\MI}T_P(f,X_i)$. Therefore, we compute the weight spectrum of $\bigcup_{f\in\MI}T_P(f,X_i)$ iteratively from the weight spectrum of $T_P(f',X_j)$ where $f'\preccurlyeq f$ and $j<i$. Additionally, the number of minimum-weight codewords in Wang-Liu shortened polar codes can be inferred from that in QUP codes by symmetry.

\subsection{Notations and Definitions}

We introduce some notations that will be employed throughout the paper. For a vector $\bm{u}\in\FF_2^N$, denote $\bm{u}_i^j = (u_i, u_{i+1}, \cdots, u_j)$. Let $\bm{u}_e = (u_j)_{j \text{ is even}}$ and $\bm{u}_o = (u_j)_{j \text{ is odd}}$ be the even and odd part of $u$ respectively. For a positive integer set $X$, denote $X_{\hat{o}} = \{\frac{i+1}{2}\mid i\in X, i \text { is odd}\}$ and $X_{\hat{e}} = \{\frac{i}{2}\mid i\in X, i \text { is even}\}$. 

Let $\mathcal{C}(m,\bm{u}_1^{i})=\{\bm{c}|\bm{c} = \tilde{\bm{u}}_1^N\bm{F}_N, \tilde{\bm{u}}_1^i=\bm{u}_1^i\}$ be the polar coset of $\bm{u}_1^{i}$. We define $\mathcal{C}_P(m,\bm{u}_1^{i},X)=\{\bm{c}_{[N]/X}|\bm{c} = \tilde{\bm{u}}_1^N\bm{F}_N, \tilde{\bm{u}}_1^i=\bm{u}_1^i\}$ as the polar coset of $\bm{u}_1^{i}$ under puncture pattern $X$. Note that if $X = [N]$, i.e., all bits are punctured, then there are $2^{m-i}$ empty codewords with length 0 and weight 0 in $\mathcal{C}_P(m,\bm{u}_1^{i},[N])$. Similarly, we define $\mathcal{C}_S(m,\bm{u}_1^{i},Y)=\{\bm{c}_{[N]/Y}|\bm{c} = \tilde{\bm{u}}_1^N\bm{F}_N , \tilde{\bm{u}}_1^i=\bm{u}_1^i, \bm{c}_Y=\bm{0} \}$ as the polar coset of $\bm{u}_1^{i}$ under shorten pattern $Y$. Note that when $Y$ complies with binary domination \cite{Jang2019}, if there exists $j \leq i, j \in Y$ such that $\bm{u}_j = 1$, then $C_N^{(i)}(\bm{u}_1^{i-1},\bm{u}_i,Y)$ is empty since the codewords in the coset contradict the shorten pattern. Let $A_d(m,\bm{u}_1^{i},X)$ and $A_d(m,\bm{u}_1^{i},Y)$ be the number of weight-$d$ codewords in $\mathcal{C}_p(m,\bm{u}_1^{i},X)$ and $\mathcal{C}_S(m,\bm{u}_1^{i},Y)$, respectively. 

\section{Number of minimum-weight codewords in the bit-reversal shortened polar codes}

In this section, we calculate the number of minimum-weight polynomials in the $r$-$th$ order decreasing monomial code $C_S(\MI,Y_i')$ with bit-reversal shortening pattern $Y_i'$. Since shortening does not change the Hamming weight of the codewords, the number of minimum-weight codewords is exactly $\sum_{f\in\MI_r}|T_S(f,Y_i')|$.

We first propose a framework to calculate the number of minimum-weight codewords in general shortened polar codes. This analysis will also be applied to Wang-Liu shortened polar codes.

From Theorem \ref{minweight2016}, the codewords in $T(f)$ can be represented as (\ref{eq:min_weight}). When the monomial $g$ is shortened, a codeword $p$ is valid if and only if the coefficient of $g$ in $p$ is 0.  This imposes an equation on the relevant coefficients $a_{j,s}$. When considering multiple shortened monomials, a system of equations emerges. In general, it is challenging to solve such a complicated system. Nevertheless, within the context of bit-reversal shortening, the equation system is indeed solvable.

\begin{example}
Assume $N=32$, $f=x_3x_5$, the complete codewords in $T(f)$ are
\begin{align*}
p = & (x_5+a_{5, 4}x_4 + a_{5, 2}x_2 + a_{5, 1}x_1 +  a_{5, 0}) \\
\cdot & (x_3 + a_{3, 2}x_2 + a_{3, 1}x_1 +  a_{3, 0}).
\end{align*}

For bit-reversal shortening, the first shortened bit corresponds to the monomial $1$, so $p$ is valid if and only if $a_{5,0}a_{3,0} = 0$, i.e., polynomials with $a_{5,0} = a_{3,0} = 1$ are excluded from the code. Therefore, $\lambda_f/4$ minimum-weight codewords are removed due to shortening, and the number of remaining codewords is $3\lambda_f/4$.

Next, the monomial $x_5$ is shortened, which requires $a_{5,0} = 0$.  Consequently, polynomials with $a_{5,0} = 1$ are excluded from the code.  Since the case $a_{5,0} = a_{3,0} = 1$ is already shortened, the newly removed polynomials are those with  $a_{5,0} = 1, a_{3,0} = 0$. Therefore, $\lambda_f/4$ more minimum-weight codewords are removed, and the number of remaining codewords with minimum weight in $T(f)$ is $\lambda_f/2$ after 2-bit bit-reversal shortening.

\end{example}

The next lemma highlights an important order property of the bit-reversal shortening.

\begin{lemma}\label{lemma_biv_seq}
Let $f,g \in \mM$, ev$(f)$ be the $q'_i$-$th$ row of $\bm{F}_N$, ev$(g)$ be the $q'_j$-$th$ row of $\bm{F}_N$. If one of the following is true:

I) $g$ is a factor of $f$;

II) $f \preccurlyeq g$ and deg$(f)=$  deg$(g)$;

then $i\leq j$, so $g$ is shortened before $f$ in bit-reversal shortening.
\end{lemma}

\begin{proof}
Denote $f = x_{i_1}\cdots x_{i_r} = x_1^{a_1}\cdots x_m^{a_m} $ and $g = x_{j_1}\cdots x_{j_t} = x_1^{b_1}\cdots x_m^{b_m}$. If $g$ is a factor of $f$, then $b_i=1$ implies $a_i=1$. Hence, $i = D(a_m,\cdots,a_1) \leq D(b_m,\cdots,b_1)= j$.

If $f \preccurlyeq g$ and deg$(f)=$ deg$(g) = r$, then $i_s \leq j_s$ for all $s\in\{1,\dots,r\}$. Therefore, $i = D(a_m,\dots, a_1) = \sum_{s=1}^{r} 2^{i_s-1} + 1 < \sum_{s=1}^{r} 2^{j_s-1} + 1 = D(b_m,\cdots,b_1)= j $. 
\end{proof}

Next, we analyse the minimum-weight codewords that are removed exactly after the $i$-$th$ bit is shortened, i.e., the set $T_S(f,Y'_{i-1})/T_S(f,Y'_i)$.

\begin{lemma}\label{lemma1}
Let $f = x_{i_1}\cdots x_{i_r} = x_1^{a_1}\cdots x_m^{a_m}$ with $F_f =  \{i_1,...,i_r\}$ and $i_1<\dots <i_r$, $g = x_{j_1}\cdots x_{j_t} = x_1^{b_1}\cdots x_m^{b_m}$ with $F_g = \{j_1,...,j_t\}$ and $j_1<\dots <j_t$. For bit-reversal shortening, suppose that $g$ is the $i$-$th$ shortened monomial. 

If $g$ is not a factor of $f$, $T_S(f,Y'_{i-1}) = T_S(f,Y'_i)$.

If $g$ is a factor of $f$, the newly removed polynomials are those of the form 
$$
p = \prod_{k\in F_f} (x_{k} + \sum_{s\in B(f,j)} a_{k,s}x_s + a_{k,0}).
$$
with $a_{k,0} = 1$ for all $k\in F_f/F_g$ and $a_{k,0} = 0$ for all $k\in F_g$. Therefore, the number of newly removed polynomials from $T(f)$ is $\lambda_f/2^r$, which means $|T_S(f,Y'_{i-1})/T_S(f,Y'_i)| = \lambda_f/2^r$.
\end{lemma}

\begin{proof}

If $g\not\preccurlyeq f$, by Theorem \ref{minweight2016}, $g$ never appears in any polynomial from $T(f)$, so shortening $g$ will not influence the polynomials in $T(f)$.

Therefore, we only need to consider $g \preccurlyeq f$, which means $t\leq r$ and $i_{t-l} \leq j_{r-l}$ for all $l\in\{1,...,t\}$. Consider the following two cases:

1) $g$ is not a factor of $f$, i.e., $F_g\not\subset F_f$. If $t=r$, according to Lemma \ref{lemma_biv_seq}, $f$ is shortened before $g$. Therefore, we only need to consider the case where $t<r$. We will prove that $g$ does not appear in any polynomial $p\in T_S(f,Y'_{i-1})$. Thus, shortening $g$ does not remove polynomials from $T_S(f,Y'_{i-1})$, i.e., $T_S(f,Y'_{i-1}) = T_S(f,Y'_i)$. 

Now, assume for contradiction that there exists some $p\in T_S(f,Y'_{i-1})$ such that the coefficient of $g$ in $p$ is $1$. Since $g \preccurlyeq f$ and $t<r$, $g$ is shortened before $f$ if and only if there exists some $s\in\{1,...,t\}$ such that $i_{t-s} < j_{r-s}$ and $i_{t-l} = j_{r-l}$ for $l<s$. Then $F_g\cap B(f, t-s) = \varnothing$, so $a_{i_{t-s},k}x_k$ for $k\in B(f, t-s)$ do not appear in $g$. Therefore, the coefficient of $g$ must be a multiple of $a_{i_{t-s}, 0}$. Denote the coefficient of $x_{i_{t-s}}g$ to be $\tilde{a}$, then the coefficient of $g$ is $a_{i_{t-s}, 0} \tilde{a}$. Since $a_{i_{t-s}, 0} \tilde{a} = 1$ implies $\tilde{a} = 1$, the coefficient of $x_{i_{t-s}}g$ is 1. However, according to Lemma \ref{lemma_biv_seq}, $x_{i_{t-s}}g$ has been shortened before $g$, which is a contradiction. 

2) $g$ is a factor of $f$, which means $F_g\subset F_f$. Then the coefficient of $g$ in $p$ is $ \prod_{s\in F_f/F_g} a_{s,0}$. Define $\beta_f(i)$ to be the number of factors of $f$ in the last $i$ bits of the bit reversal sequence $\bm{q}'$, which is the number of factors that have been shortened. Then $g$ is $\beta_f(i)$-$th$ shortened factor of $f$. We prove this by induction on $\beta_f(i)$. When $\beta_f(i) = 1$, i.e., $g=1$ and $F_g = \varnothing$, the shortened polynomials are those whose coefficient of monomial 1 is non-zero, i.e., $ \prod_{s\in F_f} a_{s,0}=1$. Therefore, the shortened polynomials must satisfy $a_{s,0}=1$ for all $s\in F_f$, and the total number of removed polynomials is $\lambda_f/2^r$. 

For the inductive step $\beta_f(i) - 1$ to $\beta_f(i)$, if $p\in T_S(f,Y'_{i-1})/T_S(f,Y'_i)$, the coefficient of $g$ in $p$ is 1, which means $\prod_{s\in F_f/F_g} a_{s,0} = 1$. From Lemma \ref{lemma_biv_seq}, any factor of $g$ has been shortened before $g$. Let $x_1^{a'_1}\cdots x_m^{a'_m}$ of $g$ with $S = \{1\leq k\leq m\mid a'_k = 1\}\subset F_g$ be a factor of $g$, and of course, a factor of $f$, whose coefficient in $p$ is $\prod_{s\in F_f/S} a_{s,0}$. Thus, polynomials with $a_{s,0} = 1$ for $s\in F_f/S$ and $a_{s,0} = 0$ for $s\in S$ have been shortened. The only remaining polynomials to be shortened are those with $a_{j,0} = 1$ for all $j\in F_f/F_g$ and $a_{j,0} = 0$ for all $j\in F_g$. Therefore, the number of newly shortened polynomials is exactly $\lambda_f/2^r$.

\end{proof}

\begin{remark}
From Lemma \ref{lemma1}, for a degree-$r$ polynomial $f$, when the monomial $g$ is shortened, the polynomials in $T(f)$ are removed from the code if and only if $g$ is a factor of $f$, and the number of removed polynomials is always $\lambda_f/2^r$.

For example, when $N=32$, $f=x_3x_5$, the polynomials in $T(f)$ are removed when monomials $1, x_5, x_3$ are shortened, and the remaining number of polynomials after each shortening step is $96, 64, 32$, respectively.
\end{remark}

The number of minimum-weight codewords in bit-reversal shortened polar codes immediately follows from Lemma \ref{lemma1}.

\begin{theorem}\label{thm_bivs}
Let $C(\MI, Y_i')$ be a shortened decreasing $r$-th order polar code with length $N=2^m$ and shortening pattern $Y_i'$, where $Y_i'$ is the last $i$ bits of the bit reversal sequence $\bm{q}'$. Let $f$ be a monomial with degree $r$ and $d=2^{m-r}$ be the minimum weight, then $|T_f(d, Y_i')| =  \lambda_f (1-\beta_f(i)/2^r)$. Here $\beta_f(i)$ is the number of factors of $f$ in $Y_i'$. Then 
\begin{align}
A_d(C(\MI, Y_i')) = \sum_{f\in\MI_r} \lambda_f (1-\beta_f(i)/2^r).
\end{align}
\end{theorem}

The Algorithm \ref{alg:1} describes the procedure for calculating the number of minimum-weight codewords for bit-reversal shortened polar codes.

\begin{figure}[!t]
\begin{algorithm}[H]
\caption{Calculate the number of minimum-weight codewords in bit-reversal shortened polar codes}
\begin{algorithmic}[1]\label{alg:1}

\renewcommand{\algorithmicrequire}{\textbf{Input:}}
\renewcommand{\algorithmicensure}{\textbf{Output:}}
\REQUIRE the $r$-order bit-reversal shortened decreasing polar code $C(\MI, Y')$.
\ENSURE the number $A$ of the minimum-weight codewords in $C(\MI, Y')$.
\STATE $A \gets 0$;
\FOR{$f = x_{i_1}\cdots x_{i_r} \in\MI_r$} 
 \STATE $\lambda_f\gets 2^{\sum_{t=1}^r (i_t-t+1)}$;
 \STATE $\hat{\lambda}_f \gets \lambda_f$;
 \FOR{$g\in Y'$}  
  \IF{$g$ is a factor of $f$}
   \STATE $\hat{\lambda}_f\gets \hat{\lambda}_f-\lambda_f/2^r$;
  \ENDIF
 \ENDFOR
 \STATE $A \gets A + \hat{\lambda}_f$;
\ENDFOR

\end{algorithmic}
\end{algorithm}
\end{figure}

\begin{remark}
The complexity of Algorithm \ref{alg:1} is $O(|Y'||\MI_r|) = O(N^2)$.
\end{remark}

\section{Number of minimum-weight codewords in QUP and Wang-Liu shortened polar codes}\label{ss33}

In this section, we calculate the number of polynomials with minimum weight in the decreasing monomial code with QUP and Wang-Liu shortening patterns. 

Due to the variability in the weight spectrum of different polynomials after puncturing, we define $P_f(d, i)$ to be the number of codewords in $T_P(f,X_i)$ with weight $d$ to facilitate a clear analysis. First, we analyse the weight spectrum of $T_P(f,X_i)$.

\subsection{Weight spectrum of $T_P(f,X_i)$}

For the sake of subsequent analysis, we begin by establishing two pivotal lemmas about the symmetry of the codewords in $T(f)$. 

\begin{lemma}\label{lemma_repeat}
For polynomial $f = x_{i_1}\cdots x_{i_t}$ with $i_1 < \dots < i_t$, $p\in T(f)$ satisfies $p _{[1,2^{i_t}]} = p_{[k2^{i_t+1}+1, (k+1)2^{i_t}]}$ for all $0\leq k\leq 2^{m-i_t}-1$. Moreover,  $\text{wt}(p _{[s,s+2^{i_t}-1]}) = \text{wt}(p)/ 2^{m-i_t}$ for all $1\leq s\leq 2^m-2^{i_t}+1$.
\end{lemma}

\begin{proof}
Note that for any $p\in T(f)$, the variable $x_j$ with $j>i_t$ does not appear in $p$. For any $1\leq l\leq 2^{i_t}$ and $0\leq k\leq 2^{m-i_t}-1$, we denote $l = D(a_1,\dots, a_m)$ and $l+k2^{i_t} =D(b_1,\dots, b_m)$. Because $a_j = b_j$ for $1\leq j\leq i_t$, it follows that $p_l = p(a_1,\dots, a_m) = p(b_1,\dots, b_m) =  p_{l+k2^{i_t}}$. Therefore,  $p_{[1,2^{i_t}]} = p_{[k2^{i_t}+1, (k+1)2^{i_t}]}$, and $\text{wt}(p_{[s,s+2^{i_t}-1]}) = \text{wt}(p)/ 2^{m-i_t}$ holds for all $1\leq s\leq 2^m-2^{i_t}+1$.
\end{proof}

\begin{lemma}\label{lemma_affine}
Let $f = x_{i_1}\cdots x_{i_t}$ with $i_1 < \dots < i_t$. We define an affine automorphism $\varphi$ that maps $x_j$ to $x_j + 1$ for all $j < i_t$ while leaving $x_j$ for $j \geq i_t$ unchanged. Then $\varphi$ satisfies that $\varphi(p)_k = p_{2^{i_t} + 1 - k}$ for all $p\in T(f)$ and $1\leq k\leq  2^{i_t}$.
\end{lemma}

\begin{proof}
Note that $\varphi$ permutes $k = D(a_1,\dots, a_{i_t}, 1, \dots, 1)$ to $ D(a_1\oplus 1,\dots,a_{i_t}\oplus 1, 1, \dots, 1) = 2^{i_t} + 1 - k$. Therefore, $\varphi(p)_k = p_{2^{i_t} + 1 - k}$ for all $p\in T(f)$ and $1\leq k\leq 2^{i_t}$.
\end{proof}

Next, we present the iterative equations of $P_f(d,a)$. The computation of $P_f(d,a)$ is based on iterating from $P_g(d,b)$ with $g \preccurlyeq f, b\leq a$.

\begin{lemma}\label{lemma2}
For a monomial $f = x_{i_1}\cdots x_{i_t}$ with $i_1<\dots < i_t, t\geq 2$, we define $N_f(w,a) = |\{p\in T(f)\mid  \text{wt}(p_{[1,a]}) = w\}|$ as the number of codewords with weight-$w$ in the first $a$ bits. Clearly, $P_f(w,a) = N_f(2^{m-t}-w,a)$. We have 

\begin{subequations}  
\begin{numcases}{N_f(w,a) =} 
\notag
\sum_{s\in B(f,t)} 2^{\alpha_f(s)} N_{f^{(s)}}(w,a) + N_{f^{(0)}}(w,a) \label{eq2a} \\
\qquad \text{ if } a\leq 2^{i_t-1}, w\neq 0;\\ 
\notag
\lambda_f - \sum_{w=1}^{2^{m-t}}  N_f(w,a) \label{eq2d} \\
\qquad \text{ if } a\leq 2^{i_t-1}, w=0;\\ 
\notag
N_f(2^{i_t-t}-w, 2^{i_t}-a), \\ 
\qquad \text{ if } 2^{i_t-1} < a \leq 2^{i_t}; \label{eq2b} \\
\notag
N_f(w-2^{i_t-t},a-2^{i_t}), \\
\qquad \text{ if } a> 2^{i_t}. \label{eq2c}
\end{numcases} 
\end{subequations}

Here, $\alpha_f(s) = |\{1\leq j\leq t-1\mid i_j>s\}|$ is the number of elements in $\{i_1,\dots, i_{t-1}\}$ that are larger than $s$, $f^{(0)} = x_{i_1}\cdots x_{i_{t-1}}$ and $f^{(s)} = x_{i_1}\cdots x_{i_{t-1}}x_s$, where $s \in B(f,t)$. 
\end{lemma}

\begin{proof}

Let $p$ be a polynomial in $T(f)$. Consider the following three cases:

1) $a > 2^{i_t}$. From Lemma \ref{lemma_repeat}, wt$(p_{[a-2^{i_t} + 1, a]}) = 2^{i_t-t}$. Therefore $\text{wt}(p_{[1,a]}) = w$ if and only if $\text{wt}(p_{[1,a-2^{i_t}]}) = w-2^{i_t-t}$, i.e.,
$$ 
N_f(w,a) = N_f(w-2^{i_t-t},a-2^{i_t}).
$$

2) $2^{i_t-1} < a \leq 2^{i_t}$. From Lemma \ref{lemma_affine}, there exists some automorphism $\varphi$, such that $p' = \varphi(p)$ and it satisfies $p'_k = p_{2^{i_t} + 1 - k}$ for $1\leq k\leq 2^{i_t}$. Therefore, $wt(p_{[1,a]}) = wt(p'_{[2^{i_t}+1-a,2^{i_t}]}) =  2^{i_t-t}-wt(p'_{[1,2^{i_t}-a]})$. Since $\varphi$ is bijective, the number of codewords in $T(f)$ with $w$ ones in the first $a$ bits is equal to the number of codewords in $T(f)$ with $2^{i_t-t}-w$ ones in the first $2^{i_t}-a$ bits, which leads to
$$ 
N_f(w,a) = N_f(2^{i_t-t}-w, 2^{i_t}-a).
$$

3) $a\leq 2^{i_t-1}$. Denote
$$
p = \prod_{j=1}^t (x_{i_j} + \sum_{k\in B(f,j)} a_{i_j, k}x_k + a_{i_j,0}).
$$
Define
\begin{align*}
\rho(p) = & (1 + \sum_{k\in B(f,t)} a_{i_t, k}x_k + a_{i_t,0}) \cdot \\
& \prod_{j=1}^{t-1} (x_{i_j} + \sum_{k\in B(f,j)} a_{i_j, k}x_k + a_{i_j,0}).
\end{align*}
Then $p_{[1,a]} = \rho(p)_{[1,a]}$ since $x_{i_t}=1$ for the first $2^{i_t-1}$ bits. Next, we will demonstrate that $\rho(p)\in T(g)$ for a suitable $g \preccurlyeq f$. The precise choice of $g$ depends on the value of $a_{i_t, k}$ for $k \in B(f,t)$. This allows us to compute $N_f(w,a)$ iteratively. Consider the following three cases:

I) If $a_{i_t, k} = 0$ for all $k\in B(f,t)$ and $a_{i_t,0} = 1$, then $\rho(p)=0$, so $wt(p_{[1,a]}) = wt(\rho(p)_{[1,a]})  = 0$. 

II) If $a_{i_t, k} = 0$ for all $k\in B(f,t)$ and $a_{i_t,0} = 0$, then $\rho(p) = \prod_{j=1}^{t-1} (x_{i_j} + \sum_{k\in B(f,j)} a_{i_j, k}x_k + a_{i_j,0}) \in T(f^{(0)})$. And every polynomial $\rho(p)$ in $T(f^{(0)})$ corresponds to one polynomial $p = x_{i_t}\rho(p)$ in $T(f)$. Therefore, there are $ N_{f^{(0)}}(w,a)$ codewords in $T(f)$ with $w$ ones in the first $a$ bits when $a_{i_t, k} = 0$ for all $k\in B(f,t)$ and $a_{i_t,0} = 0$.

III) If $a_{i_t, k} = 1$ for some $k\in B(f,t)$, then $\rho(p)\in T(f^{(s_1)})$, where $s_1 \in B(f,t)$ is the largest integer for which $a_{i_t,s_1} = 1$. In fact, several polynomials in $T(f)$ correspond to the same $p'$ in $T(f^{(s_1)})$, and this set is denoted as $H_{p'}$. Specifically, $H_{p'} \triangleq \{p\in T(f)\mid \rho(p) = p' \} = \{p\in T(f)\mid p_{[1, 2^{i_t-1}]} =  (p')_{[1, 2^{i_t-1}]} \}$. Note that for any $p' \in T(f^{(s_1)})$ and $p'' \in T(f^{(s_2)})$, $(p')_{[1, 2^{i_t-1}]} \neq (p'')_{[1, 2^{i_t-1}]}$. Otherwise, by Lemma \ref{lemma_repeat}, $p'=p''$ since $s_1,s_2 \leq i_t-1$, which is a contradiction. Therefore, 
\begin{equation}\label{eq5.1}
H_{p'} \cap H_{p''} = \varnothing.
\end{equation}

Next, we will prove that $|H_{p'}| = 2^{\alpha_f(s_1)}$ for all $s_1 \in B(f,t)$ and $p'\in T(f^{(s_1)})$.

Let $h = x_{s_1} + \sum_{k\in B(f,t), k<s} a_{i_t, k}x_k + a_{i_t,0} + 1$ and $p' = h\prod_{j=1}^{t-1} (x_{i_j} + \sum_{k\in B(f,j)} a_{i_j, k}x_k + a_{i_j,0})\in T(f^{(s_1)})$. Now, for $i_j>s_1$, $h (x_{i_j} + \sum_{k\in B(f,j)} a_{i_j, k}x_k + a_{i_j,0}) = h (x_{i_j} + \sum_{k\in B(f,j)} a_{i_j, k}x_k + a_{i_j,0} + h + 1)$. Thus, for all $\beta_j\in\{0,1\}$, $h\prod_{i_j<s}(x_{i_j} + \sum_{k\in B(f,j)} a_{i_j, k}x_k + a_{i_j,0}) \prod_{i_j>s, j\leq t-1}(x_{i_j} + \sum_{k\in B(f,j)} a_{i_j, k}x_k + a_{i_j,0} + \beta_j(h+1))$ are equal to $p'$ in the first $2^{i_t-1}$ bits, so 
\begin{equation}\label{eq5.2}
|H_{p'}| \geq 2^{\alpha_f(s_1)}.
\end{equation}

We have already shown there are at least $2^{\alpha_f(s_1)}$ polynomials in $T(f)$ equal to $p'$ in the first $2^{i_t-1}$ bits. Next we claim that $|H_{p'}|$ is exactly $2^{\alpha_f(s_1)}$. For $s_1, s_2\in B(f,t)$, assume $s_2<s_1$ and there does not exist $s_2<s_3<s_1$ such that $s_3\in B(f,t)$. Then $\alpha_f(s_1) + s_1 = \alpha_f(s_2) + s_2  + 1$. Furthermore, if $s_0$ is the minimum integer in $B(f,t)$, then $\alpha_f(s_0) + s_0 = t$. Hence, if $s$ is the $k$-$th$ smallest integer in $B(f,t)$, $\alpha_f(s) + s = t +k - 1$. Since $|B(f,t)| = i_t-t$, we have
\begin{equation}\label{eq5.3}
\sum_{s\in B(f,t)} 2^{\alpha_f(s)+s-t+1} = \sum_{i=1}^{i_t-t}2^i =  2^{i_t-t+1}-2.
\end{equation}

Then
\begin{align*}
\lambda_f - 2\lambda_{f^{(0)}} &\geq \sum_{s\in B(f,t)} \sum_{p'\in T(f^{(s)})} |H_{p'}| \\
& \geq \sum_{s\in B(f,t)} 2^{\alpha_f(s)} \lambda_{f^{(s)}} \\
&= \sum_{s\in B(f,t)} 2^{\alpha_f(s)+s-t+1} \lambda_{f^{(0)}} = \lambda_f - 2\lambda_{f^{(0)}}.
\end{align*}
Here, the first inequality is from (\ref{eq5.1}). Note that the numbers of polynomials in case I) and II) are both $\lambda_{f^{(0)}}$, which needs to be subtracted. The second inequality is from (\ref{eq5.2}). The equality is from multiplying both side of (\ref{eq5.3}) by $\lambda_{f^{(0)}}$. Therefore, $|H_{p'}| = 2^{\alpha_f(s_1)}$ for all $s_1 \in B(f,t)$ and $p'\in T(f^{(s_1)})$.

In conclusion, for $w\neq 0$, $N_f(w,a)$ is equal to the sum of the number of polynomials in case II and case III, which is
$$
\sum_{s\in B(f,t)} 2^{\alpha_f(s)} N_{f^{(s)}}(w,a) + N_{f^{(0)}}(w,a),
$$
and $N_f(0,a)$ is the number of the remaining polynomials.
\end{proof}

\begin{example}

Table \ref{tab3} shows $P_f(w,a)$ for $N=32$, $f = x_2x_3$, $4\leq w\leq 8$ and $0\leq a \leq 16$. 

\begin{table*}[htbp] 
\begin{center} 
\setlength{\tabcolsep}{1.5pt}{
\begin{tabular}{c|c|c|c|c|c|c|c|c|c|c|c|c|c|c|c|c|c} 
\diagbox{weight}{puncturing bits} & 0 & 1 & 2 & 3 & 4 & 5 & 6 & 7 & 8 & 9 & 10 & 11 & 12 & 13 & 14 & 15 & 16\\
\hline 
4 & 0 & 0 & 0 & 0 & 0 & 0 & 0 & 0 & 0 & 0 & 1 & 2 & 4 & 6 & 9 & 12 & 16\\
\hline 
5 & 0 & 0 & 0 & 0 & 0 & 0 & 0 & 0 & 0 & 4 & 6 & 8 & 8 & 8 & 6 & 4 & 0\\
\hline 
6  & 0 & 0 & 1 & 2 & 4 & 6 & 9 & 12 & 16 & 12 & 9 & 6 & 4 & 2 & 1 & 0 & 0\\
\hline 
7 & 0 & 4 & 6 & 8 & 8 & 8 & 6 & 4 & 0 & 0 & 0 & 0 & 0 & 0 & 0 & 0 & 0\\
\hline 
8 & 16 & 12 & 9 & 6 & 4 & 2 & 1 & 0 & 0 & 0 & 0 & 0 & 0 & 0 & 0 & 0 & 0\\

\end{tabular}}
\caption{The number of codewords with different weights and the number of puncturing bits}
\label{tab3}
\end{center}
\end{table*}

For $8\leq a\leq 16$, we have $P_f(w,a) = N_f(8-w,a) = N_f(6-w,a-8) = P_f(w+2,a-8)$, a result from (\ref{eq2c}). Additionally, we note that $P_f(w,a) = P_f(14-w,8-a)$ when $4\leq a\leq 8$ according to (\ref{eq2b}). 

For $1\leq a \leq 4$ and $w\neq 8$, $P_f(w,a)$ is calculated by (\ref{eq2a}). In fact, we know that a minimum-weight codeword in $T(x_2x_3)$ can be represented as 
$$
p=(x_3+a_{3,1}x_1 + a_{3,0})(x_2+a_{2,1}x_1 + a_{2,0}).
$$
When $1\leq a \leq 4$, $x_3=1$, so $p_{[1,4]} = p'_{[1,4]}$, where
$$
p' =(a_{3,1}x_1 + a_{3,0} + 1)(x_2+a_{2,1}x_1 + a_{2,0}).
$$

We consider three cases for $p'$:

Case I: $a_{3,1} = 0, a_{3,0} = 1$, then $p' = 0$. 

Case II: $a_{3,1} = a_{3,0} = 0$, then $p' = x_2+a_{2,1}x_1 + a_{2,0} \in T(x_2)$. 

Case III: $a_{3,1} = 1$, then $p' = (x_1 + a_{3,0} + 1)(x_2+a_{2,1}x_1 + a_{2,0}) \in T(x_2x_1)$. Note that $(x_1 + a_{3,0} + 1)(x_2+a_{2,1}x_1 + a_{2,0})  = (x_1 + a_{3,0} + 1)(x_2+a_{2,1}(a_{3,0}+1) + a_{2,0})$. One polynomial $p'$ in $T(x_2x_1)$ corresponds to two polynomials in $T(x_1)$.

In conclusion, we have $N_f(w,a) = N_{x_2}(w,a) + 2N_{x_1x_2}(w,a)$ for $1\leq a \leq 4$, and $P_f(w,a) =  N_f(8-w,a)$. Finally, $P_f(8,a) =  N_f(0,a) = 16 - \sum_{w=1}^{8}  N_f(w,a)$ according to (\ref{eq2d}). 
\end{example}

\subsection{Number of minimum-weight codewords in QUP polar codes}

Building upon Lemma \ref{lemma2}, we derive the weight spectrum of $T_P(f, X_i)$ after puncturing. Next, we prove that a minimum-weight codeword after puncturing must also be a minimum-weight codeword at the mother code length, that is, it belongs to  $T_P(f, X_i)$. 

Let ${\bm f}_N^{(j)}$ be the $j$-$th$ row vector of $\bm{F}_N$. Note that the minimum weight of $\mathcal{C}(m,(\bm{0}_1^{j-1},1))$ is equal to wt$({\bm f}_N^{(j)})$ \cite[Corollary 1]{Li2019}. Also, note that $\text{wt}({\bm f}_{N/2}^{(j)}) = \text{wt}({\bm f}_N^{(j)})  = \text{wt}({\bm f}_N^{(j+N/2)})/2 $ for $1 \leq j \leq N/2$.

\begin{lemma}\label{lemma_dmin}
$({\bm f}_N^{(j)})_{[N]/X_i}$ is a minimum-weight codeword of $\mathcal{C}(m,(\bm{0}_1^{j-1},1),X_i)$. Moreover, the minimum-weight codewords of $\mathcal{C}(m,(\bm{0}_1^{j-1},1),X_i)$ are also minimum-weight codewords at the mother code length. 
\end{lemma}

\begin{proof}
The proof is by induction on $m$. The base case, $m=1$, can be proved directly. For the inductive step, from $m-1$ to $m$, we consider two cases:

1) $1 \leq j \leq N/2$, which means ${\bm f}_N^{(j)}$ is in the top half of $\bm{F}_N$. Then $\mathcal{C}(m,(\bm{0}_1^{j-1},1),X_i) = \{(\bm{u}_{[N/2]/X_i} \oplus \bm{v}_{[N/2]/X_i}, \bm{v}) \mid \bm{u}\in \mathcal{C}(m-1,(\bm{0}_1^{j-1},1)),  \bm{v}\in \FF_2^{N/2} \}$. Then by the induction hypothesis,
\begin{align}\label{lemma_dmin_eq1}
\notag
& \text{wt}((\bm{u}_{[N/2]/X_i} \oplus \bm{v}_{[N/2]/X_i}, \bm{v})) \\
\notag
=  & \text{wt}(\bm{v}_{X_i}) + \text{wt}((\bm{u}_{[N/2]/X_i} \oplus  \bm{v}_{[N/2]/X_i}, \bm{v}_{[N/2]/X_i})) \\
\geq & \text{wt}(\bm{u}_{[N/2]/X_i}) \geq \text{wt}(({\bm f}_{N/2}^{(j)})_{[N/2]/X_i}) =  \text{wt}(({\bm f}_{N}^{(j)})_{[N]/X_i}).
\end{align}
Therefore, $({\bm f}_N^{(j)})_{[N]/X_i}$ is the minimum-weight codeword. 

Furthermore, if $(\bm{u}_{[N/2]/X_i} + \bm{v}_{[N/2]/X_i}, \bm{v})$ is a minimum-weight codeword, then the equality in (\ref{lemma_dmin_eq1}) must hold. By inductive hypothesis, $\text{wt}(\bm{u}_{[N/2]/X_i}) =  \text{wt}(({\bm f}_{N/2}^{(j)})_{[N/2]/X_i})$ implies that $\bm{u}$ is the minimum-weight codeword at the mother code length $N/2$. In addition, $\text{wt}((\bm{u}_{[N/2]/X_i} \oplus  \bm{v}_{[N/2]/X_i}, \bm{v})) = \text{wt}(\bm{u}_{[N/2]/X_i})$ implies that $\text{wt}(\bm{v}_{X_i}) = 0$ and $\bm{v}_{[N/2]/X_i}$ is covered by $\bm{u}_{[N/2]/X_i}$, i.e., $v_k = 1$ implies $u_k = 1$ for any $k\in [N/2]/X_i$. Hence, $\bm{v}$ is covered by $\bm{u}$. Then $\text{wt}((\bm{u} \oplus \bm{v}, \bm{v})) = \text{wt}(\bm{u}) = \text{wt}({\bm f}_{N/2}^{(j)}) =  \text{wt}({\bm f}_{N}^{(j)})$, so $(\bm{u} \oplus  \bm{v}, \bm{v})$ is the minimum-weight codeword at the mother code length $N$.

2) $N/2 + 1 \leq j \leq N$, i.e., ${\bm f}_N^{(j)}$ is in the lower half of $\bm{F}_N$. If $i>N/2$, the codewords can be regarded as those punctured from mother code length $N/2$, then the lemma follows from inductive hypothesis. If $i\leq N/2$, $\mathcal{C}(m,(\bm{0}_1^{j-1},1),X_i) = \{(\bm{v}_{[N/2]/X_i}, \bm{v})\mid \bm{v}\in \mathcal{C}(m-1,(\bm{0}_1^{j-N/2-1},1)) \}$. Then via induction,
\begin{align}\label{lemma_dmin_eq2}
\notag
\text{wt}((\bm{v}_{[N/2]/X_i}, \bm{v})) & \geq \text{wt}((({\bm f}_{N/2}^{(j-N/2)})_{[N/2]/X_i}, {\bm f}_{N/2}^{(j-N/2)})) \\
& = \text{wt}(({\bm f}_{N}^{(j)})_{[N]/X_i}).
\end{align}
Therefore, $({\bm f}_N^{(j)})_{[N]/X_i}$ is the minimum-weight codeword. 

Furthermore, if $(\bm{v}_{[N/2]/X_i}, \bm{v})$ is the minimum-weight codeword, then the equality in (\ref{lemma_dmin_eq2}) must hold, that is, $\text{wt}(\bm{v}) =  \text{wt}({\bm f}_{N/2}^{(j)})$. Then $(\bm{v}, \bm{v})$ is the minimum-weight codeword at the mother code length $N$.

In conclusion, the minimum-weight codewords in $C(\MI, X_i)$ are also minimum-weight codewords at the mother code length, meaning they belong to some $T_P(f, X_i)$. 
\end{proof}

The number of minimum-weight codewords in QUP polar codes follows directly from Lemma \ref{lemma2} and Lemma \ref{lemma_dmin}.

\begin{theorem}\label{thm_qup}
Let $C(\MI,X_i)$ be a punctured decreasing $r$-th order polar code with length $N=2^m$ and puncturing pattern $X_i$, where $X_i$ is the first $i$ bits of the sequence $\bm{q}$. Denote $d$ to be the minimum positive integer such that $\sum_{f\in\MI} P_f(d,X_i) > 0$.  Then $d$ is the minimum weight of $C(\MI,X_i)$ and $A_d(C_P(\MI, i)) = \sum_{f\in\MI} P_f(2^{m-r}-d,X_i)$ and $A_w(C_P(\MI, i)) \geq \sum_{f\in\MI} P_f(2^{m-r}-w,X_i)$ for $w>d$.
\end{theorem}

The Algorithm \ref{alg:2} concludes the procedure for calculating the number of minimum codewords in QUP polar code. It is worth noting that the lower bound $\sum_{f\in\MI} P_f(2^{m-r}-w,i)$ in Theorem \ref{thm_qup} for $w$ greater than the minimum weight can also be calculated using Lemma \ref{lemma2}.

\begin{figure}[!t]
\begin{algorithm}[H]
\caption{Calculate the number of minimum-weight codewords in QUP polar code}
\begin{algorithmic}[1]\label{alg:2}

\renewcommand{\algorithmicrequire}{\textbf{Input:}}
\renewcommand{\algorithmicensure}{\textbf{Output:}}
\REQUIRE the $r$-order QUP decreasing polar code $C(\MI, X_i)$.
\ENSURE the minimum weight $d$ and the number $A$ of the minimum codeword in $C(\MI, X_i)$.
\STATE Denote $N_f(w,a)$ to be the number of codewords in $T(f)$ with $w$ ones in the first $a$ bits, initialize $N_f(w,a) \gets 0$;
\FOR{Let $f$ be a monomial indexed from $N-1$ to $0$}
 \STATE Denote $f = x_{i_1}\cdots x_{i_t}$;
 \IF{$t=1$} 
  \STATE Calculate $N_f(w,a)$ by exhaustive search for all codewords in $T(f)$;
  \STATE \textbf{Continue}
 \ENDIF
 \FOR{$a = 0$ to $|X_i|$} 
  \FOR{$w = 1$ to $2^{m-t}$} 
   \IF{$0\leq a \leq 2^{i_t-1}$} 
    \STATE $N_f(w,a) \gets \sum_{j\in B(f,t)} 2^{\alpha_f(j)} N_{f_j}(w,a) + N_{f_0}(w,a)$; 
   \ENDIF
   \IF{$2^{i_t-1} < a\leq 2^{i_t}$} 
    \STATE $N_f(w,a) \gets N_f(2^{i_t-t}-w, 2^{i_t}-a)$;  
   \ENDIF
   \IF{$a > 2^{i_t} + 1$} 
    \STATE $N_f(w,a) \gets N_f(w-2^{i_t-t},a-2^{i_t})+ 2^{i_t-t}$; 
   \ENDIF
  \ENDFOR
  \STATE $N_f(0,a) \gets \lambda_f - \sum_{w=1}^{2^{m-t}}  N_f(w,a)$;
 \ENDFOR
\ENDFOR
\STATE Denote $d$ to be the minimum positive integer such that $\sum_{f\in\MI_r} N_f(2^{m-r}-d,X) > 0$;
\STATE $P_f(d,i) \gets N_f(2^{m-r}-d,i)$
\STATE $A\gets \sum_{f\in \MI_r} P_f(d,i)$;
\end{algorithmic}
\end{algorithm}
\end{figure}

\begin{remark}
As the code length $N$ approaches infinity, the time complexity of Algorithm \ref{alg:2} is $O(N^3)$.
\end{remark}

\begin{remark}
Since an LTA transformation is an automorphism of decreasing polar codes, applying this transformation to the puncturing pattern will result in the same weight spectrum. Consequently, we can determine the weight spectrum for a series of punctured decreasing polar codes. 
\end{remark}

\subsection{Number of minimum-weight codewords in Wang-Liu shortened polar codes}

The determination of the number of minimum-weight codewords in the Wang-Liu shortened polar codes can be directly inferred from the analysis developed for QUP codes. 

\begin{theorem}\label{thm5}
Let $f$ be a codeword with $wt(f) = 2^{m-r}$. Then $|T_S(2^{m-r},Y_i)| = P_f( 2^{m-r}, i)$ for all $1\leq i\leq N$. 
\end{theorem}

\begin{proof}
For Wang-Liu shortening, a codeword $\bm{c}$ in $T(f)$ is not shortened if and only if $\bm{c}_{[2^m-i,2^m-1]}=0$. In other words, it must retain its original weight $2^{m-r}$ after the last $a$ bits punctured.

The theorem is a direct consequence of the symmetry: puncturing the last $i$ bits yields the same weight spectrum as puncturing the first $i$ bits through the affine transformation $x_t\to x_t +1$ for all $1\leq t\leq m$.
\end{proof}

\begin{theorem}\label{thm_qus}
Let $C(\MI,Y_i)$ be a shortened decreasing $r$-th order polar code with shortening pattern $Y_i$ and minimum distance $d=2^{m-r}$, where $Y_i$ is the last $i$ bits of the QU sequence $\bm{q}$. Then $A_d(C(\MI, Y_i)) = \sum_{f\in\MI_r} P_f(d,i)$.
\end{theorem}

\begin{example}
Let $C(\MI,Y_2)$ be the code with a length of $N=8$, the information set $\MI=\{x_1x_2,x_2\}$ and last two bits shortened.  Then $d_{\min}(C(\MI,Y_2)) = 2$. According to Theorem \ref{thm5} and \ref{thm_qus}, $A_2(C(\MI,Y_2)) =  P_{x_1x_2}(2,2) = 2$, where $P_{x_1x_2}(2,2)$ can be determined using Algorithm \ref{alg:2}.
\end{example}

\section{Average weight spectrum of pre-transformed rate-compatible polar codes}
In this section, we introduce efficient formulas for calculating the average weight spectrum of rate-matched pre-transformed polar codes. A distinguishing attribute of our approach is its versatility-it operates without the limitations imposed by partial ordering or specific rate matching patterns. Furthermore, our approach enables the calculation of the complete weight spectrum. The key idea of our approach lies in mitigating the complexity of the enumeration problem through averaging over a suitably chosen ensemble of codes. Specifically, we derive the average weight spectrum by taking the expectation over a set of random pre-transformation matrices with $i.i.d.$ upper triangular components. This strategy significantly reduces the number of unique terms to a linear scale, thus circumventing the exponential computational complexity typically associated with calculating specific weight spectrum. 

To begin our analysis, we provide a recursive computation of the weight spectrum for polar cosets. For clarity, we break down the coset into its odd and even components. This allows for a decomposition for the length $N$ code as a direct sum of two codes \cite[$\S 9$ of ch.2]{MacWilliams1977}, each of half the original length $\frac{N}{2}$. This approach is similar to the technique employed in \cite{Yao2023} for calculating the spectrum of coset codes, albeit with a focus on the mother code length. Here, we extend this method to accommodate rate-matching.

\begin{theorem}
\label{thm7}
\begin{figure*}
\begin{align}
\label{eq4}
&A_d(m+1,\bm{u}_1^{2i},X) = \sum_{d_1+d_2=d}A_{d_1}(m,\bm{u}_{1,e}^{2i},X_{\hat{e}})A_{d_2}(m,\bm{u}_{1,o}^{2i} \oplus  \bm{u}_{1,e}^{2i},X_{\hat{o}})
\end{align}
\begin{align}
\label{eq5}
&A_d(m+1,\bm{u}_1^{2i-1},X)  =  A_d(m+1,(\bm{u}_1^{2i-1},0),X)+A_d(m+1,(\bm{u}_1^{2i-1},1),X) \notag \\
& = \sum_{d_1+d_2=d}A_{d_1}(m,(\bm{u}_{1,e}^{2i-2},0),X_{\hat{e}})A_{d_2}(m,(\bm{u}_{1,o}^{2i-2} \oplus  \bm{u}_{1,e}^{2i-2},u_{2i-1}),X_{\hat{o}}) \notag \\
& + \sum_{d_1+d_2=d}A_{d_1}(m,(\bm{u}_{1,e}^{2i-2},1),X_{\hat{e}})A_{d_2}(m(\bm{u}_{1,o}^{2i-2} \oplus  \bm{u}_{1,e}^{2i-2},u_{2i-1}+ 1),X_{\hat{o}})
\end{align}
\end{figure*}
The recursive formulas for the weight spectrum of punctured polar cosets are presented in Equations (\ref{eq4}) and (\ref{eq5}). The shortening case follows the same recursive procedure but the puncturing pattern $X$ is replaced by the shortening pattern $Y$. The only difference lies in the boundary conditions when $m=0$, i.e., the code length is equal to 1. If the code bit is neither punctured nor shortened,
\begin{align*}
&A_0(0,u_1 = 0,\varnothing) = 1, A_1(0,u_1 = 0, \varnothing) = 0,\\
&A_0(0,u_1 = 1, \varnothing) = 0, A_1(0,u_1 = 1, \varnothing) = 1. 
\end{align*}
If the code bit is punctured,
\begin{align*}
&A_0(0,u_1 \in \{0,1\},X=\{1\}) = 1,\\ 
&A_1(0,u_1 \in \{0,1\},X=\{1\}) = 0. 
\end{align*}
If the code bit is shortened,
\begin{align*}
&A_0(0,u_1 = 0,Y=\{1\}) = 1, A_1(0,u_1 = 0,Y=\{1\}) = 0, \\
&A_0(0,u_1 = 1,Y=\{1\}) = A_1(0,u_1 = 1,Y=\{1\})= 0.
\end{align*}

\end{theorem} 

Since the proof of Theorem \ref{thm7} is similar as \cite[Theorem 1]{Yao2023}, we will only illustrate the key points of the proof.

\begin{proof}
Drawing from the proof in \cite[Theorem 1]{Yao2023}, it is established that the odd and even components of a coset code are independent, i.e., $(\bm{u}\bm{F}_N)_o = (\bm{u}_o + \bm{u}_e)\bm{F}_{N/2} = \bm{u}' \bm{F}_{N/2}$ and $(\bm{u}\bm{F}_N)_e = \bm{u}_e\bm{F}_{N/2}$, where $\bm{u}'$ and $\bm{u}_e$ are independent of each other. This independence implies that the weight spectrum of a coset of length-$N$ can be determined through the spectrum of two constituent cosets, each of length-$\frac{N}{2}$. In comparison to the mother code length case, the distinctive feature lies in the application of rate-matching patterns to the half-length cosets. Specifically, for the two cosets of length-$\frac{N}{2}$, their respective rate-matching patterns correspond to the odd and even components extracted from the original length-$N$ coset.

As for the boundary condition, it is straightforward for the case without rate-matching. For the puncturing case, $\mathcal{C}_P(0,\bm{u}_1,X = \{1\})$ is the coset with one empty codeword, so $A_0(0,\bm{u}_1 ,X=\{1\}) = 1$ and $A_1(0,\bm{u}_1,X=\{1\}) = 0$. For the shortening case, $\bm{u}_1 = 1$ contradicts with the shortening pattern $Y=1$, then $\mathcal{C}_S(0,\bm{u}_1 = 1,Y = \{1\})$ is empty, so $A_0(0,\bm{u}_1 = 1,Y=\{1\}) = A_1(0,\bm{u}_1 = 1,Y=\{1\})= 0$. $\mathcal{C}_S(0,\bm{u}_1 = 1,Y = \{1\})$ contains one empty codeword, so $A_0(0,\bm{u}_1 = 0,Y=\{1\}) = 1, A_1(0,\bm{u}_1 = 0,Y=\{1\}) = 0$.

\end{proof}

\begin{example}
Suppose $m=3$, puncturing pattern $X = \{1,2,5\}$, consider the weight spectrum of coset $\mathcal{C}_P(3,(0,0,1,0),X)$. From the proof of Theorem \ref{thm7}, we decompose it into odd and even parts, that is, $\mathcal{C}_P(3,(0,0,1,0),X) = \{(\bm{c}_o, \bm{c}_e)|\bm{c}_o \in \mathcal{C}_P(2, (0, 1) \oplus (0,0), X_{\hat{o}} = \{1,3\}),  \bm{c}_e \in \mathcal{C}_P(2, (0,0), X_{\hat{e}}=\{1\})\}$. Then $A_d(3,(0,0,1,0),X)  =   \sum_{d_1+d_2=d} A_d(2,(0,1),X_{\hat{o}})A_d(2,(0,1),X_{\hat{e}})$.

\end{example}

The number of codewords with weight $d$ for pre-transformed codes under puncturing pattern $X$ (shortening pattern $Y$) with transformation matrix $\tilde{\bm{T}}$ is denoted by $N(d,\tilde{\bm{T}},X)$ ($N(d,\tilde{\bm{T}},Y)$). The average number is denoted by $E[N(d,\bm{T},X)]$ ($E[N(d,\bm{T},Y)]$), where the expectation is with respect to random pre-transformation matrix $\bm{T}$, and $\bm{T}_{ij}$, $1\leq i < j \leq N$ are $i.i.d.$  $Bernoulli(\frac{1}{2})$ $r.v.$. The average weight spectrum can be calculated through the weight spectrum of cosets.

\begin{theorem}
Let $\mathcal{I} = \{I_1,...,I_K\}$ be the information set, $I_j^s = |\{i|Y_i>I_j\}|$ be the number of shorten bits with indices larger than $I_j$. Then
\begin{align}
E[N(d,\bm{T},X)] = \sum_{1 \leq j \leq K}2^{K-j}\frac{A_d(m,(\bm{0}_1^{I_j-1},1),X)}{2^{N-I_j}}
\end{align}

\begin{align}
E[N(d,\bm{T},Y)] = \sum_{1 \leq j \leq K}2^{K-j}\frac{A_d(m,(\bm{0}_1^{I_j-1},1),Y)}{2^{N-I_j-I_j^s}}
\end{align}  
\end{theorem}

\begin{proof}
Let ${\bm f}_N^{(i)}$  be the $i$-$th$ row vector of $\bm{F}_N$ and ${\bm g}_N^{(i)}$ be the $i$-$th$ row vector of $\bm{G}_N = \bm{T}\bm{F}_N$ after puncturing, where $\bm{T}$ is the random pre-transformation matrix. The proof is similar to \cite[Lemma 1]{Li2021}. The only difference is that for punctured codes,
$$
P \left(\text{wt}\left({\bm g}_N^{(I_j)}\right)=d \right) = \frac{A_d(m,(\bm{0}_1^{I_j-1},1),X)}{2^{N-I_j}}, 
$$
since the number of codewords in $\mathcal{C}_P(m,(\bm{0}_1^{I_j-1},1),X)$ is $2^{N-I_j}$. And for shortened codes,
$$
P \left(\text{wt}\left({\bm g}_N^{(I_j)}\right)=d \right) = \frac{A_d(m,(\bm{0}_1^{I_j-1},1),Y)}{2^{N-I_j-I_j^s}},
$$
since the number of codewords in $\mathcal{C}_S(m,(\bm{0}_1^{I_j-1},1),Y)$ is $2^{N-I_j-I_j^s}$.
\end{proof}

The algorithm to calculate the weight spectrum of punctured pre-transformed polar cosets is described in Algorithm \ref{alg:4}. When considering shortening, the underlying algorithm remains identical. The only divergence between the two cases lies in the boundary conditions.

\begin{algorithm*}
\caption{Calculate the weight spectrum of punctured polar cosets}
\begin{algorithmic}[1]\label{alg:4}

\renewcommand{\algorithmicrequire}{\textbf{Input:}}
\renewcommand{\algorithmicensure}{\textbf{Output:}}
\REQUIRE the punctured polar code $C_P(\MI,X)$ with code length $N=2^m$, dimension $K$ and information set $\MI =\{I_1,\dots, I_K\}$.
\ENSURE the average weight spectrum $E[N(d,\bm{T},X)]$ with $d = w_{\min}, \dots , 2^m$, where $w_{\min}$ is the minimum row weight of information bits.

\STATE Initialize $A_d(m,0,X)$ and $A_d(m,1,X)$ as Theorem \ref{thm7};  

\FOR{$m_1 = 1; m_1\leq m, m_1++$}
\FOR{$i = 1; i\leq 2^{m_1}, i++$}
\IF{$i$ is even}
  \FOR{$d=0$ to $2^m$}
    \STATE $A_d(m, \bm{0}_1^i, X) = \sum_{d_1+d_2=d}A_{d_1}(m-1, \bm{0}_1^{i/2}, X_{\hat{e}}) A_{d_2}(m-1, \bm{0}_1^{i/2}, X_{\hat{o}})$;
    \STATE $A_d(m, (\bm{0}^{i-1}, 1), X) = \sum_{d_1+d_2=d}A_{d_1}(m-1, (\bm{0}^{i/2-1}, 1), X_{\hat{e}}) A_{d_2}(m-1, (\bm{0}^{i/2-1}, 1),X_{\hat{o}})$;
  \ENDFOR
\ELSE
  \FOR{$d=0$ to $2^m$}
    \STATE $A_d(m, (\bm{0}_1^{i-1}, 1), X) = \sum_{d_1+d_2=d}A_{d_1}(m-1,\bm{0}_1^{(i+1)/2},X_{\hat{e}}) A_{d_2}(m-1,(\bm{0}_1^{(i-1)/2}, 1),X_{\hat{o}}) + \sum_{d_1+d_2=d}A_{d_1}(m-1,(\bm{0}_1^{(i-1)/2}, 1),X_{\hat{e}}) A_{d_2}(m-1,\bm{0}_1^{(i+1)/2},X_{\hat{o}})$;  
    \STATE $A_d(m, \bm{0}_1^i, X) = \sum_{d_1+d_2=d}A_{d_1}(m-1,\bm{0}_1^{(i+1)/2},X_{\hat{e}}) A_{d_2}(m-1,\bm{0}^{(i+1)/2} ,X_{\hat{o}})+ \sum_{d_1+d_2=d}A_{d_1}(m-1,(\bm{0}_1^{(i-1)/2},1),X_{\hat{e}}) A_{d_2}(m-1 ,(\bm{0}_1^{(i-1)/2},1),X_{\hat{o}})$;
  \ENDFOR
\ENDIF
\ENDFOR
\ENDFOR
\STATE $ E[N(d,\bm{T},X)] = \sum_{1 \leq j \leq K}2^{K-j}\frac{A_d(m,(\bm{0}_1^{I_j-1},1),X)}{2^{N-I_j}}. $

\end{algorithmic}
\end{algorithm*}

\begin{remark}
Let $\chi(N)$ denote the complexity of Algorithm \ref{alg:4}. Then $O(N)$ operations are required for computing each $A_d(m, \bm{0}_1^i, X)$ and $A_d(m, (\bm{0}_1^{i-1}, 1), X)$ with $0\leq i\leq 2^m, 0\leq d\leq 2^m$ after the computation of average spectrum with code length $\frac{N}{2}$, so $\chi(N) = \chi(N/2) + O(N^3)$. Consequently, $\chi(N) = O(N^3)$. 
\end{remark}

\section{Numerical Results}

In this section, we present numerical results for the weight spectrum of rate-compatible polar codes.

Table \ref{tab1} compares the value of $\sum_{f\in\MI} P_f(w,X)$ calculated in Theorem \ref{thm_qup} with the weight spectrum collected through SCL decoding with a large list \cite{Li2012}, for weights less than or equal to $2^{m-r}$. The results demonstrate that the lower bound in Theorem \ref{thm_qup} can be accurate when the number of punctured bits is small. However, when the number of punctured bits is large, the lower bound might be significantly smaller than the actual value. The reason is that a greater number of bits punctured can cause codewords with weights exceeding $2^{m-r}$ to be rendered with weights less than or equal to $2^{m-r}$.

\begin{table}[htbp] 
\begin{center} 
\setlength{\tabcolsep}{1.5pt}{
\begin{tabular}{c|c|c|c} 
\hline
$(E,K)$ & $d$ & $\sum_{f\in\MI} P_f(w,X)$ & $A_d$
 \\
\hline
\multirow{3}*{(112,64)} & 6 & 168 & 168\\
\cline{2-4} & 7 & 352 & 352\\
\cline{2-4} & 8 & 210 & 210\\
\hline
\multirow{4}*{(108,64)} & 5 & 12 & 12\\
\cline{2-4} & 6 & 250 & 250\\
\cline{2-4} & 7 & 324 & 324\\
\cline{2-4} & 8 & 150 & 354\\
\hline
\multirow{5}*{(104,64)} & 4 & 4 & 4\\
\cline{2-4} & 5 & 16 & 16\\
\cline{2-4} & 6 & 356 & 356\\
\cline{2-4} & 7 & 256 & 256\\
\cline{2-4} & 8 & 195 & 911\\
\hline
\end{tabular}}
\caption{The estimated and accurate number of low-weight codewords in QUP polar codes}
\label{tab1}
\end{center}
\end{table}

Figures \ref{fig_weight1} to \ref{fig_weight8} compare the union bounds \cite{Sason2006} calculated by the weight spectrum obtained through our methods with the performance of SCL decoding with list size 32 for QUP, Wang-Liu shortened and bit-reversal shortened polar codes with different lengths and rates. Here the union bounds are calculated as $\sum_{d \leq  2^{m-r}}A_d Q(-\sqrt{d}/\sigma)$ where $\sigma$ is the standard variance of additive white Gaussian noise and $Q(x) = \frac{1}{\sqrt{2\pi}}\int_x^{\infty} e^{-z^2/2} dz$ is the spectrum function of the standard normal distribution. The information set is chosen by the Gaussian Approximation \cite{Trifonov2012} under rate matching. For shortened codes, the numbers of minimum-weight codewords are derived from Theorem \ref{thm_bivs} and Theorem \ref{thm_qus}. For $r$-$th$ order punctured codes, the numbers of codewords with weight no larger than $2^{m-r}$ are approximated based on Theorem \ref{thm_qup}. The  minimum weight $d_{\min}$ and the number of minimum-weight codewords $A_{d_{\min}}$ are presented in Table \ref{tab2}. The simulation results show that the union bound calculated by our methods closely matches the performance of SC-list decoding, particularly at high SNR.  

\begin{figure}[!t]
\centering
\includegraphics[width=0.5\textwidth, trim = 80 220 80 220, clip]{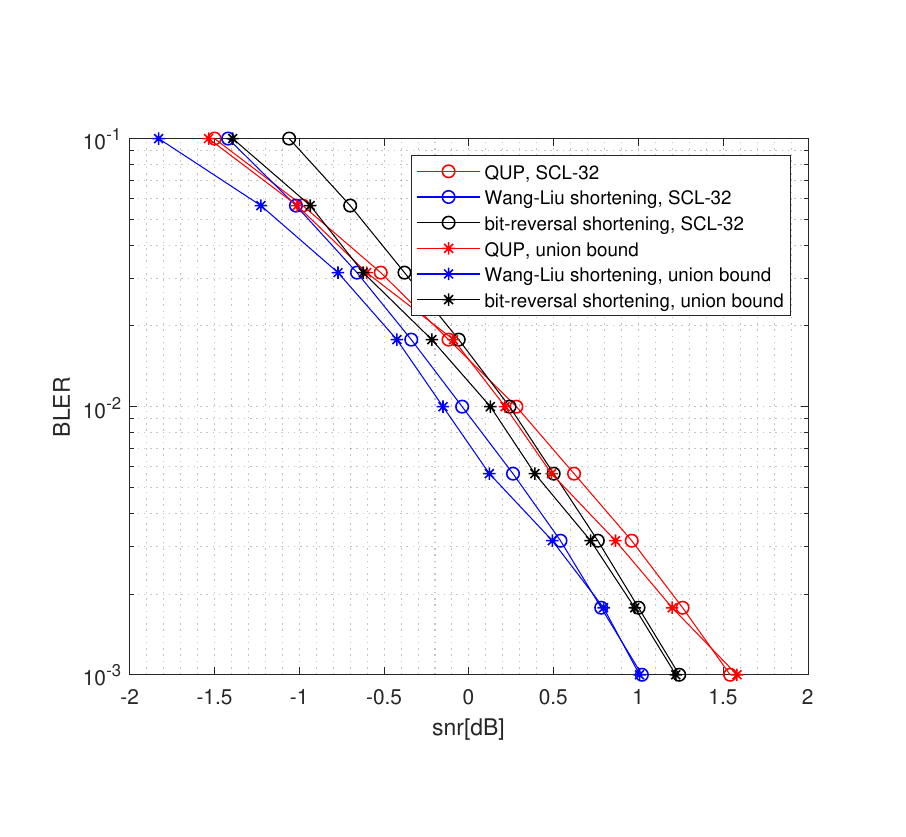}
\caption{SCL performance and union bound for different [96,24] polar codes}
\label{fig_weight1}
\end{figure}

\begin{figure}[!t]
\centering
\includegraphics[width=0.5\textwidth, trim = 80 220 80 220, clip]{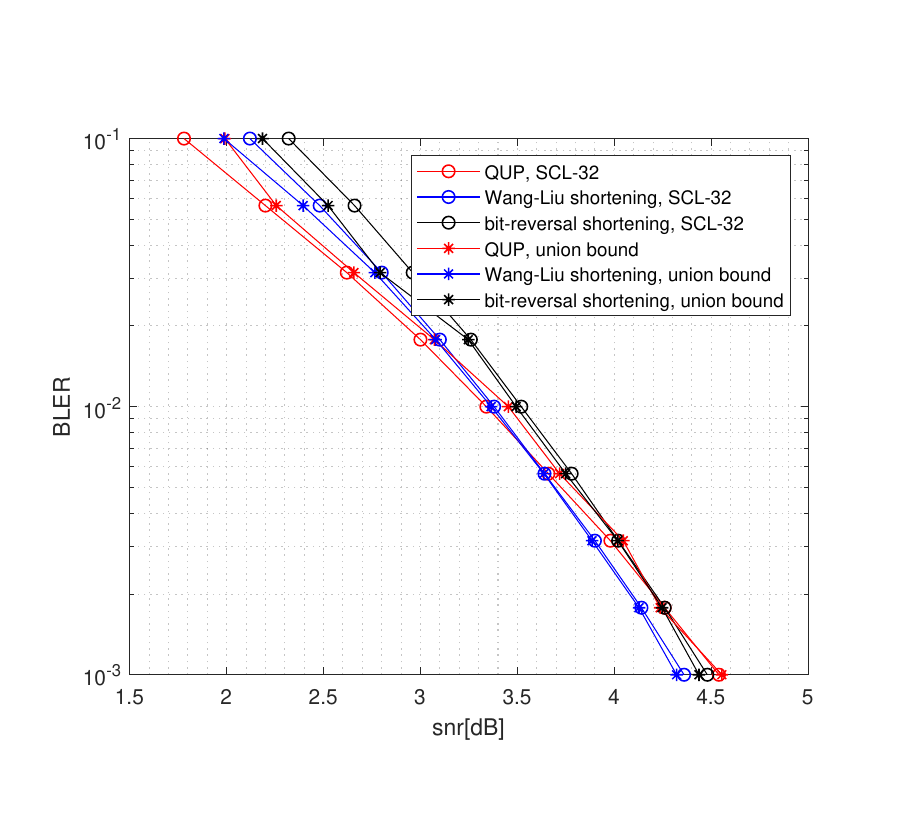}
\caption{SCL performance and union bound for different [96,48] polar codes}
\label{fig_weight2}
\end{figure}

\begin{figure}[!t]
\centering
\includegraphics[width=0.5\textwidth, trim = 80 220 80 220, clip]{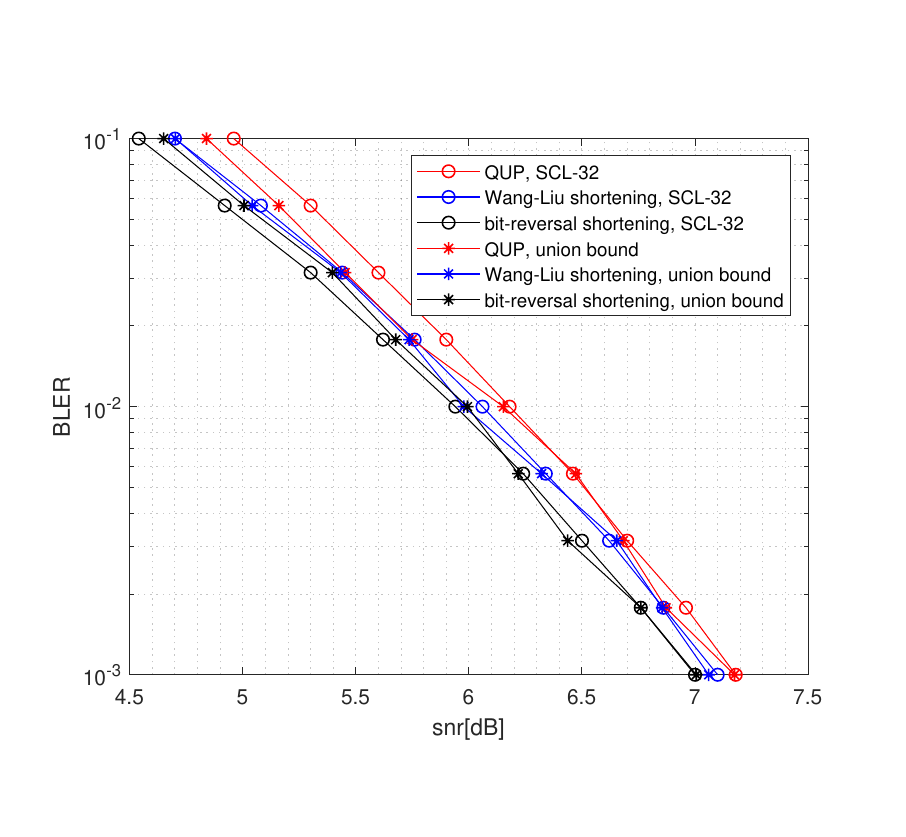}
\caption{SCL performance and union bound for different [96,72] polar codes}
\label{fig_weight3}
\end{figure}

\begin{figure}[!t]
\centering
\includegraphics[width=0.5\textwidth, trim = 80 220 80 220, clip]{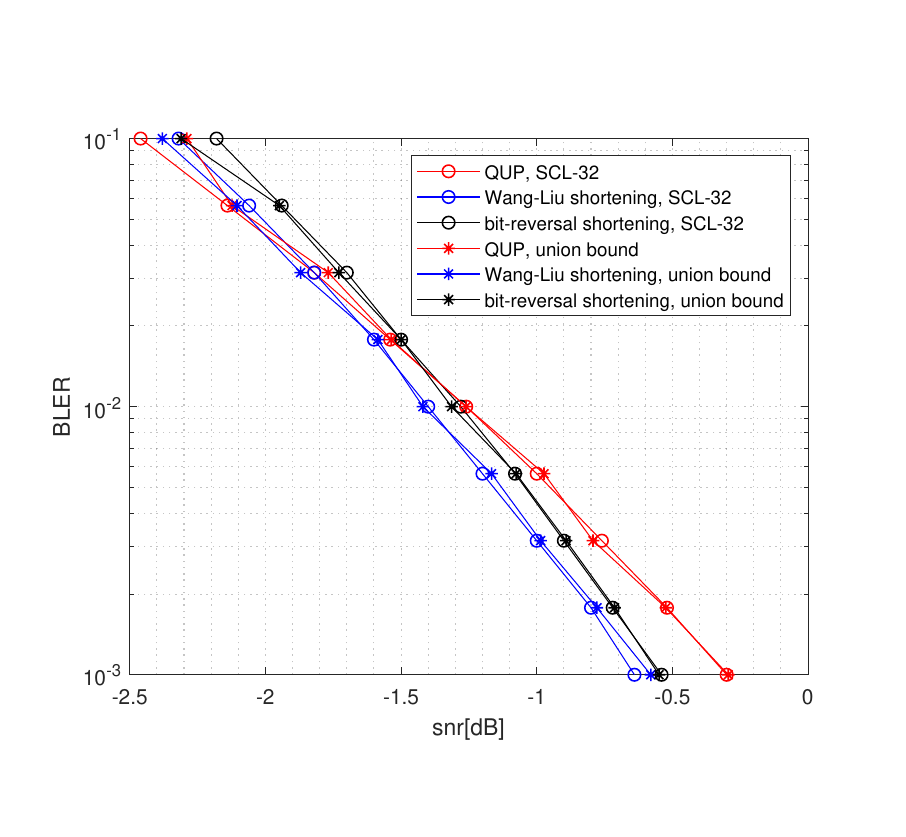}
\caption{SCL performance and union bound for different [768,192] polar codes}
\label{fig_weight6}
\end{figure}

\begin{figure}[!t]
\centering
\includegraphics[width=0.5\textwidth, trim = 80 220 80 220, clip]{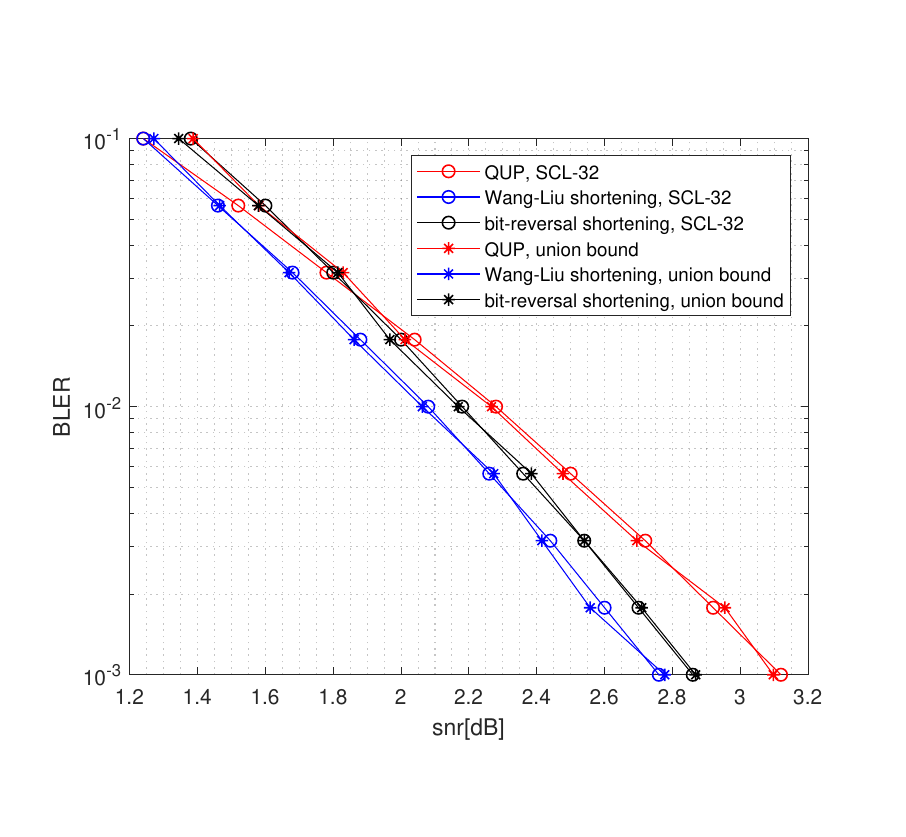}
\caption{SCL performance and union bound for different [768,384] polar codes}
\label{fig_weight7}
\end{figure}

\begin{figure}[!t]
\centering
\includegraphics[width=0.5\textwidth, trim = 80 220 80 220, clip]{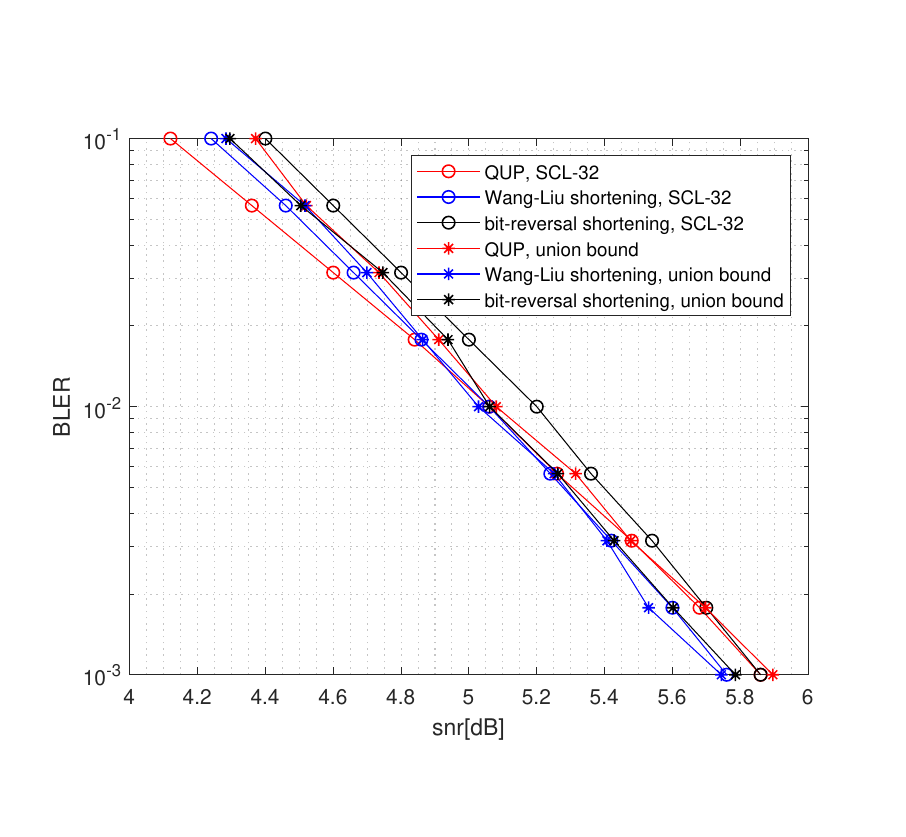}
\caption{SCL performance and union bound for different [768,576] polar codes}
\label{fig_weight8}
\end{figure}

\begin{table*}[htbp] 
\begin{center} 
\setlength{\tabcolsep}{1.5pt}{
\begin{tabular}{c|c|c|c|c} 
\hline
\multicolumn{2}{c|}{$(E,K)$} & quasi-uniform puncturing & Wang-Liu shortening & bit-reversal shortening
 \\
\hline
\multirow{2}*{(96,24)} & $d_{\min}$ & 12 & 16 & 16\\
\cline{2-5} & $A_{d_{\min}}$ & 56 & 292 & 490\\
\hline
\multirow{2}*{(96,48)} & $d_{\min}$ & 6 & 8 & 8\\
\cline{2-5} & $A_{d_{\min}}$ & 48 & 648 & 900\\
\hline
\multirow{2}*{(96,72)} & $d_{\min}$ & 4 & 4 & 4\\
\cline{2-5} & $A_{d_{\min}}$ & 392 & 336 & 264\\
\hline
\multirow{2}*{(786,192)} & $d_{\min}$ & 24 & 32 & 32\\
\cline{2-5} & $A_{d_{\min}}$ & 864 & 13456 & 18248\\
\hline
\multirow{2}*{(786,384)} & $d_{\min}$ & 12 & 16 & 16\\
\cline{2-5} & $A_{d_{\min}}$ & 2752 & 50464 & 72080\\
\hline
\multirow{2}*{(786,576)} & $d_{\min}$ & 6 & 8 & 8\\
\cline{2-5} & $A_{d_{\min}}$ & 896 & 47168 & 70944\\
\hline
\end{tabular}}
\caption{The minimum weight and minimum-weight codeword number for different codes}
\label{tab2}
\end{center}
\end{table*}

Figure \ref{fig_weight4} and \ref{fig_weight5} compare the union bounds \cite{Sason2006} calculated by the average weight spectrum obtained through our methods with the performance of SCL decoding with list size 32 for rate-matching PC-polar codes. The parity check is implemented by a 5-length cyclic shift register \cite{Zhang2018}, and all frozen bits are used for parity check. The information set is chosen according to 5-$th$ generation New Radio (5G-NR) communications systems \cite{TS38.212}. The rate-matching patterns include QUP, Wang-Liu shortening and sub-block interleaving in 5G-NR communications systems. The simulation results indicate that the union bound derived from our methods also closely approximates the performance of pre-transformed polar codes, particularly at high SNR.  

\begin{figure}[!t]
\centering
\includegraphics[width=0.5\textwidth, trim = 80 220 80 220, clip]{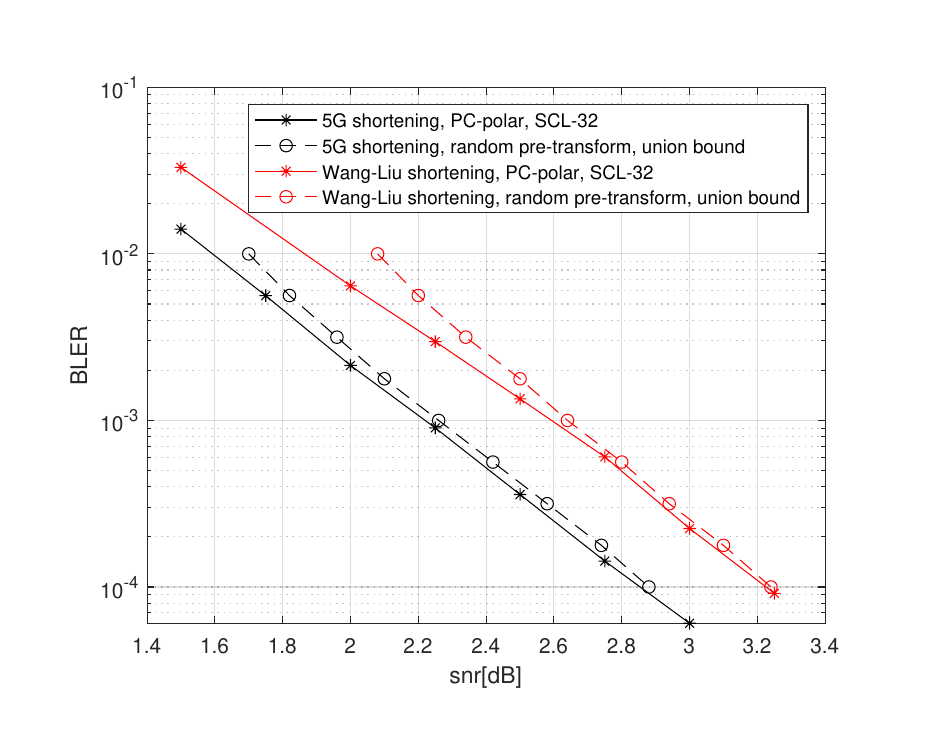}
\caption{SCL performance and union bound for different pre-transformed [896,448] polar codes}
\label{fig_weight4}
\end{figure}

\begin{figure}[!t]
\centering
\includegraphics[width=0.5\textwidth, trim = 80 220 80 220, clip]{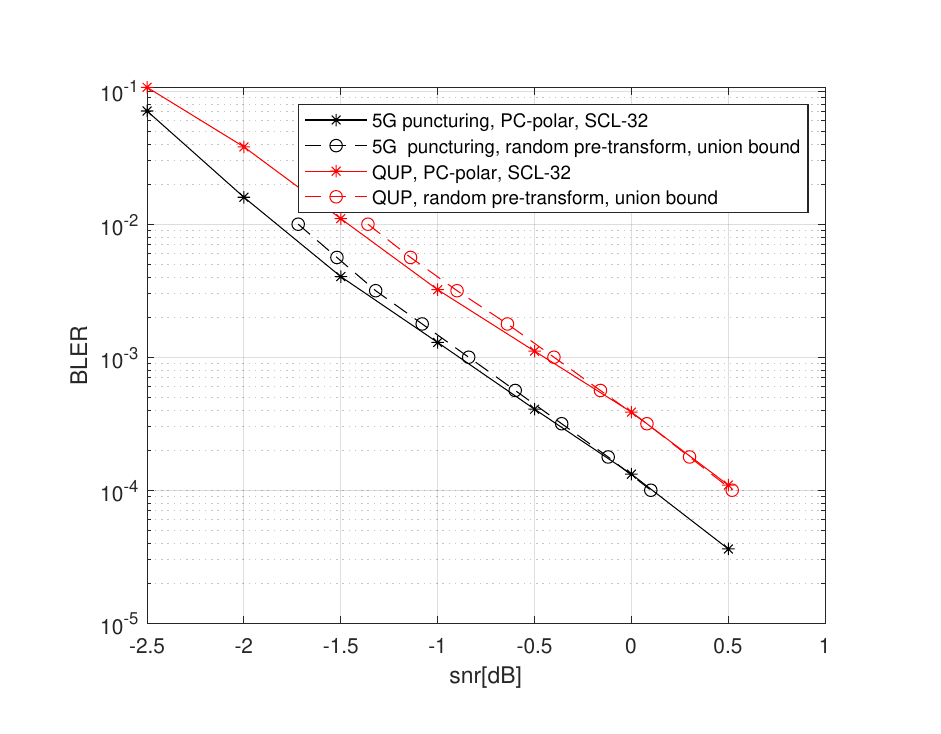}
\caption{SCL performance and union bound for different pre-transformed [640,160] polar codes}
\label{fig_weight5}
\end{figure}

\section{Conclusion}

In this paper, we develop a systematic framework for the enumeration of the minimum-weight codewords in QUP, Wang-Liu shortened, and bit-reversal shortened decreasing polar codes. Furthermore, we compute the average weight spectrum of pre-transformed polar codes, applicable to more expansive rate-matching configurations. Iterative formulas and algorithms are provided with a polynomial complexity with respect to code length. Empirical simulations show that our proposed techniques are capable of approximating the error-correcting performance of rate-matching polar codes with large list size.

\end{document}